\documentclass[12pt]{article}
\usepackage{graphicx}
\usepackage{enumitem}
\usepackage{amsthm}
\usepackage{float}
\usepackage{array}
\usepackage{makecell} 
\usepackage{pbox} 
\usepackage{subcaption}
\usepackage{xcolor}
\usepackage{soul} 
 
\newcommand{\abs}[1]{\left\vert #1 \right\vert}
\newcommand{\pa}[1]{\left( #1 \right)}
\newcommand{\be}{\begin{equation}}
\newcommand{\ee}{\end{equation}}

\newtheorem{thm}{Theorem}[section]
\newtheorem{lem}[thm]{Lemma}
\newtheorem{rem}[thm]{Remark}

\usepackage{indentfirst} 
\usepackage{amsmath}
\usepackage{amssymb}

\usepackage{geometry}
\usepackage[utf8]{inputenc}
\usepackage[section]{placeins}
\usepackage{hyperref}
\usepackage{tabularx,ragged2e,booktabs}
\usepackage{authblk}
\usepackage{tcolorbox} 
\usepackage{algorithm} 
\usepackage[noend]{algpseudocode}

\begin{document}

\title{A spatial multiscale mathematical model of \textit{Plasmodium vivax} transmission}

\author[1]{Shoshana Elgart}
\author[2]{Mark B. Flegg}
\author[3]{Somya Mehra}
\author[3]{Jennifer A. Flegg}
\affil[1]{Laurel Springs School, Ojai, California, United States}
\affil[2]{School of Mathematics, Monash University, Melbourne Australia}
\affil[3]{School of Mathematics and Statistics, The University of Melbourne, Parkville, Australia}
\date{}                     
\setcounter{Maxaffil}{0}
\renewcommand\Affilfont{\itshape\small}
\maketitle
\begin{abstract}
The epidemiological behavior of \textit{Plasmodium vivax} malaria occurs across spatial scales including within-host, population, and metapopulation levels. On the within-host scale, \textit{P. vivax} sporozoites inoculated in a host may form latent hypnozoites, the activation of which drives secondary infections and accounts for a large proportion of \textit{P. vivax} illness; on the metapopulation level, the coupled human-vector dynamics characteristic of the population level are further complicated by the migration of human populations across patches with different malaria forces of (re-)infection. To explore the interplay of all three scales in a single two-patch model of \textit{Plasmodium vivax} dynamics, we construct and study a system of eight integro-differential equations with periodic forcing (arising from the single-frequency sinusoidal movement of a human sub-population). Under the numerically-informed ansatz that the limiting solutions to the system are closely bounded by sinusoidal ones for certain regions of parameter space, we derive a single nonlinear equation from which all approximate limiting solutions may be drawn, and devise necessary and sufficient conditions for the equation to have only a disease-free solution. Our results illustrate the impact of movement on \textit{P. vivax} transmission and suggest a need to focus vector control efforts on forest mosquito populations. The three-scale model introduced here provides a more comprehensive framework for studying the clinical, behavioral, and geographical factors underlying \textit{P. vivax} malaria endemicity.
\end{abstract}
\section{Introduction}\label{Intro}
Malaria is one of the most significant sources of morbidity and mortality across the world, resulting in nearly 250 million cases and accounting for over 60,000 deaths in 2021 alone \cite{world2022world}. The disease is particularly dangerous to young children (likely accounting for over $7\%$ of all deaths in children under five) and to persons with compromised immune systems \cite{laishram2012complexities, monroe2022reflections}. 

Global efforts towards malaria elimination have refocused attention on disease caused by \textit{Plasmodium vivax}, a malaria parasite responsible for a significant fraction of recurring blood-stage infections. \textit{P. vivax} infections are characterized by the presence of hypnozoites -- latent forms of the \textit{Plasmodium} sporozoites inoculated by the bites of infectious \textit{Anopheles} mosquitoes -- in host hepatocytes. Certain hypnozoites may \textit{activate} after initial primary infections, causing malaria relapses in which the host may both experience blood-stage infection symptoms and is liable to transmit malaria to additional mosquitoes, thus significantly increasing the number of secondary cases that a primary case may give rise to \cite{white2014modelling}. Current anti-hypnozoital pharmaceutical interventions (chiefly the 8-aminoquinoline drugs primaquine and tafenoquine) are limited in their efficacy due to the risk for erythrocyte rupture in individuals with  G6PD deficiency, emerging parasite resistance, and lack of funding for treatment availability in the private sector \cite{collins1996primaquine, frank2005diagnosis,actwatch2017private}.

In the regions of Southeast Asia where \textit{P. vivax} are highly prevalent, even small-scale spatial behavior is known to play a crucial role in disease transmission. In Cambodia, for example, where over $90\%$ of malaria incidence involves \textit{P. vivax} parasites, the majority of remaining endemic districts are isolated from urban areas and situated in proximity to tropical forests \cite{durnez2013outdoor}. Owing to high humidity and seasonal floods, these forests are typically characterized as zones of severe prevalence among \textit{Anopheles} populations, particularly when compared with nearby villages, where mosquitoes are targeted with long-lasting insecticidal nets, residual spraying, and other interventions \cite{jongdeepaisal2021acceptability}. Mobile groups (often termed ``forest-goers") who migrate between lower-transmission villages and high-transmission forests are considered to be at highest risk for contracting -- and spreading -- malaria, with the residual endemicity in forest-adjacent villages often attributed to the movement of forest-goers \cite{chhim2021malaria,durnez2013outdoor, phok2022behavioural}.  
Similar patterns of movement-based transmission have been observed in neighboring Thailand, Laos, and Vietnam \cite{phok2022behavioural}. 

A number of spatial or \textit{P. vivax} within-host mathematical models have been posed to describe malaria epidemiological dynamics and predict transmission. Since the model we present in this paper unifies the spatial and within-host representations, we discuss previous work most relevant to ours from both categories below. 

We begin with spatial malaria models (e.g., \cite{anicta2019regional, auger2008ross, gao2012multipatch, moukam2018spatial, prosper2012assessing, rodriguez2001models}) which either represent a vector-borne disease in general or specifically focuses on \textit{Plasmodium falciparum} dynamics (this simplifies the model structure by rendering the inclusion of within-host dynamics less necessary). The articles \cite{auger2008ross, rodriguez2001models, prosper2012assessing} extend the SI (susceptible-infectious) Ross-Macdonald model to several patches. In particular, Prosper et al.  \cite{prosper2012assessing} analyzes a two-patch model, and the works Auger et al. \cite{auger2008ross}, Rodríguez and Torres-Sorando \cite{rodriguez2001models} study general $n$-patch models using arbitrary-large systems of ordinary differential equations (ODEs). In contrast to the original SI model framework, some spatial models represent human and mosquito populations in SEIR (susceptible-exposed-infectious-recovered) and SEI frameworks, respectively, so as to incorporate the significant latency period in malaria infections \cite{gao2012multipatch}. Non-ODE spatial models include a two-patch reaction-diffusion system and a numerical process-based model \cite{anicta2019regional, moukam2018spatial}. Some models specify a residence patch for each individual or  focus on long-period, long-distance movement \cite{prosper2012assessing, rodriguez2001models}. A large proportion of studies construct autonomous nonlinear ODE models, which are then analyzed using traditional methods (i.e., a linear stability analysis yielding bounds or explicit values for the models' $R_0$) \cite{auger2008ross, gao2012multipatch, rodriguez2001models, prosper2012assessing}. 

Models representing \textit{P. vivax} are generally rare, since the significance of the hypnozoite reservoir suggest that such models need a within-host component. White et al. \cite{white2014modelling} introduced a large ODE model (of up to $102$ equations) with host compartments of the form $S_i, I_i$, where individuals in $S_i$ or $I_i$ carry exactly $i$ hypnozoites, and numerically simulated the model to yield results on \textit{P. vivax} epidemiology. Mehra et al. \cite{Mehra} devised a framework for modeling both short-latency tropical and long-latency temperate malaria strains in a single host, representing hypnozoite accrual, activation, and clearance, as well as primary infection dynamics, through an open network of $\cdot/M/\infty$ (infinitely-many) server queues with batch arrivals. A probability generating function (PGF) is then derived to encapsulate the within-host dynamics associated with the individual in question. A population-level model was later built around this framework in Anwar et al. \cite{Anwar}.

In this paper, we propose and analytically study the first mathematical model combining \textit{P. vivax} within-host and metapopulation dynamics. We further extend the framework devised in Anwar et al. \cite{Anwar} and Mehra et al. \cite{Mehra}, generating a two-patch model where the human population is divided into moving and non-moving individuals. Assuming that moving individuals' migration patterns align, we derive a system of eight integro-differential equations for the human-vector populations. We proceed to study the limiting equations of the system analytically using the ansatz that solutions are limit-periodic, the assumption that mean population movement is sinusoidal with rapid frequency, and the additional assumption that patch populations see very small rates of change. We structure the paper as follows: In Section \ref{model}, we introduce the model, adapt the within-host framework in Mehra et al. \cite{Mehra} to derive our system of metapopulation-level equations, and consider the constant force of reinfection case. In Section \ref{results}, we present our analytical and numerical results before providing concluding remarks in Section \ref{discussion}. 

\section{Model development} \label{model}

We introduce spatial dynamics into the multi-scale \textit{Plasmodium vivax} framework developed in \cite{Anwar}. A representative schematic of the model is presented in Figure \ref{fig1}.

\subsection{Spatial patch model}
In the general case, we distribute our collection of individuals and \textit{Anopheles} mosquitoes across an arbitrary number of $n$ distinct patches (blue hexagons). Metapopulation movement between patches takes the form of traversing a connected sub-graph of $\mathbb{Z}$ with $n$ vertices representing the patches (Figure \ref{fig1}). Each Patch $i$ contains a local population of mosquitoes, $\mathcal{Q}_i$. We denote the set of patches containing permanent human settlement -- so-called   ``village patches" below -- by  $J \subset [1, n]$. In Figure \ref{fig1} village patches are represented by a brown border around their present populations and placed on the left side of the respective patches whilst populations on non-village patches are placed on the right side of their patch with a green border. 

Every patch $i \in J$ is the permanent residence of two subclasses of  humans: moving humans ($\mathcal{M}_i$) and non-moving humans ($\mathcal{N}_i$), as shown connected to each respective patch in Figure \ref{fig1}. In particular, we define $\mathcal{M}_i$ to represent the subpopulation of humans residing on Patch $i$ who are mobile and regularly travel to a different patch and define $\mathcal{N}_i$ as the subpopulation of humans who do not travel beyond their village.  
Viewing $\mathcal{N}_i$ and $\mathcal{M}_i$ as \textit{sets} of humans we define the population \be\begin{aligned} \mathcal{H}_i := \mathcal{N}_i \cup \mathcal{M}_i, \end{aligned}\ee the set of all humans with permanent residence on Patch $i$. The proportion that  each subpopulation forms of the total Patch $i$ population is represented without calligraphic font; $0 < M_i \leq 1$ and $N_i = 1-M_i$, where 
\be \begin{aligned} M_i := \frac{\abs{\mathcal{M}_i}}{\abs{\mathcal{H}_i}}, \hspace{20pt} N_i := \frac{\abs{\mathcal{N}_i}}{\abs{\mathcal{H}_i}}.\end{aligned} \ee

This is motivated by the studies Kunkel et al. \cite{kunkel2021choosing} and Sandfort et al. \cite{sandfort2020forest}, which found that most movement between areas with high and low malaria transmission in Cambodia is dominated by a single sub-population, typically consisting of younger men working as loggers or engaged in farming activities.  

We focus most of our analysis in this paper on the simplest version of this patch model, the $n = 2$ case of a ``village--forest" model (Figure \ref{fig2}) describing the situation in which humans travel between an area of high settlement and low \textit{Anopheles} density (Patch $1$, a ``village" location) to an area of low settlement and high \textit{Anopheles} density (Patch $2$, a ``forest" location). We consider interactions between a single human population $\mathcal{H}$ with permanent residence on Patch $1$ (separated into a moving set of humans $\mathcal{M}$ and a non-moving set of humans $\mathcal{N}$) and two sets of static mosquito populations $\mathcal{Q}_1$ and $\mathcal{Q}_2$, associated with Patch $1$ and Patch $2$, respectively.

Analogously to the $n$-patch case, we define the quantities 
\be\begin{aligned}M:= \frac{\abs{\mathcal{M}}}{\abs{\mathcal{H}}}, \hspace{20pt} N:= \frac{\abs{\mathcal{N}}}{\abs{\mathcal{H}}} \\ Q_1:= \frac{\abs{\mathcal{Q}_1}}{\abs{\mathcal{H}}}, \hspace{20pt} Q_2:= \frac{\abs{\mathcal{Q}_2}}{\abs{\mathcal{H}}},\end{aligned}\ee
such that $M, N$ represent the proportion that the moving and non-moving populations, respectively, form of the total human population, while $Q_1, Q_2$ represent the sizes of the village and forest mosquito populations, respectively, relative to the human population.

The choice of qualifiers ``forest" and ``village" is motivated by studies exploring the epidemiological dynamics that result when individuals in South and Southeast Asia travel between forested areas where malaria is endemic and neighboring villages where malaria is largely suppressed \cite{kunkel2021choosing, ranjha2021forest, sandfort2020forest}.

\begin{figure}
     \centering
     \begin{subfigure}[b]{0.45\textwidth}
         \centering
         \includegraphics[width=\textwidth]{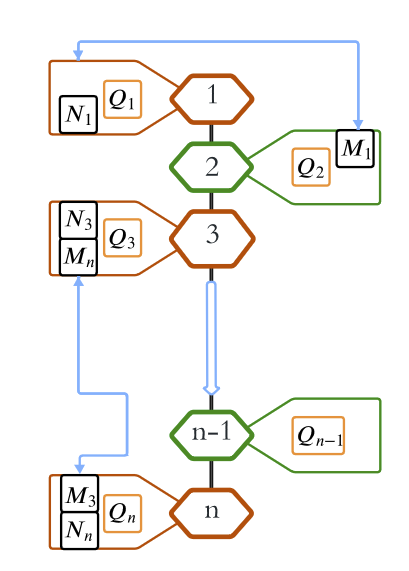}
         \caption{}
         \label{fig1}
     \end{subfigure}
    \begin{subfigure}[b]{0.5\textwidth}
         \centering
         \includegraphics[width=\textwidth]{Two-Patch.png}
         \caption{}
         \label{fig2}
     \end{subfigure}
     \caption{(a) Diagram of the $n$-patch model at a certain time $t$ (not including inter-compartment interactions). Patches $2$ and $n-1$ are ``forest patches" (outlined in green) while Patches $1$, $3$, and $n$ are ``village" patches (outlined in brown). Each patch $i$ is associated with a mosquito population $\mathcal{Q}_i$ (outlined in orange), while all village patches $j \in J \subset [1, n]$ contain a non-moving human population $\mathcal{N}_j$. Light blue arrows indicate the movement of a group of individuals within $\mathcal{M}_k, k \in J$, traveling from their patch ($k$) of permanent residence to a different patch (arrows are two-sided to indicate that moving individuals may return). (b) Diagram of the two-patch model at some time $t$. The proportions of moving individuals on Patches $1, 2$, respectively, are $A_1(t), A_2(t)$. Unlike Patch $1$, Patch $2$ does not contain \textit{permanent} human residents, and thus has no non-moving individuals. }
\end{figure} 

We make the following simplifying assumptions in our two-patch model:

\begin{enumerate}[label=(\roman*)]
\item Only humans can move between patches, while mosquitoes remain on a single patch. This assumption is motivated by Auger et al. \cite{auger2008ross} and by estimates of \textit{Anopheles} dispersal ranges, which have been found to be below a kilometer on average \cite{midega2007estimating, saddler2019development}. 

\item The moving group $\mathcal{M}$ moves deterministically at the time-dependent, periodic rates $r_1(t)$ (from Patch $1$ to Patch $2$) and $r_2(t)$ (from Patch $2$ to Patch $1$).

\item Each individual moving person is assumed to migrate in discrete jumps from one patch to the next. The sequence of transitions governing this movement is assumed to be time-periodic. We assume that each individual continues this periodicity of movement even throughout periods of malaria infection, a reasonable hypothesis in locations where asymptomatic cases dominate prevalence, such as Southeast Asia \cite{baum2016common, starzengruber2014high}.

\item We assume that the $\mathcal{Q}_1$ population is much smaller than the $\mathcal{Q}_2$ population (i.e., $Q_1 << Q_2$). Moreover, we suppose that $Q_1 + Q_2 < 1$ (from White et al. \cite{white2014modelling}, which places the mean mosquito-to-human ratio in a homogeneous location at under 0.6), and that the death rate of Patch $2$ mosquitoes is far lower than that of Patch $1$ mosquitoes. 

\item Immigration is not included, meaning that the  total number of humans is constant over time. We further assume that both $\abs{\mathcal{M}}$ and $\abs{\mathcal{N}}$ are considerable, which allows us to define the total number of individuals on a patch as a continuous function of time and assume deterministic dynamics for the population as a whole (even as individual primary infection and hypnozoite dynamics remain stochastic). 
\end{enumerate}

Under assumption (ii), we have that  \begin{equation} \label{eqbas} \begin{aligned} \frac{dA_1}{dt} = -r_1(t)A_1 + r_2(t)A_2 = - \frac{dA_2}{dt}. \end{aligned} \end{equation}

Since $A_1(t) + A_2(t) = 1$, both $A_1(t)$ and $A_2(t)$ are periodic non-constant functions for all initial conditions by Floquet's theorem \cite{dacunha2011unified}. Moreover, as $r_1(t), r_2(t) > 0$,  $A_1(t), A_2(t)$ lie between $0$ and $1$ for all $t$. We choose the initial conditions $A_1(0) = u + v, A_2(0) = 1 - u - v$, and let $r_1(t) = v\omega\sin(\omega t), r_2(t) = -v\omega\sin(\omega t)$, such that \be\label{A1A2} \begin{aligned} A_1(t) = u + v \cos\pa{\omega t}, \\ A_2(t) = 1 - u - v \cos\pa{\omega t}.\end{aligned} \ee for parameters $u, v, \omega$, where $\omega$ is in the units of days$^{-1}$. We assume that $0 < \abs{v} < < 0.5 < u < 1$ and $u + \abs{v} < 1$, which implies that the average density of moving individuals on Patch $1$ exceeds that on Patch $2$.

\subsubsection{Infection status within patches}

We briefly introduce some additional notation on the population-level, which will be referenced in the within-host analysis below. The force of reinfection (FORI), denoted $\lambda_i(t)$, is unique to a patch as it is directly dependent on the proportion of infectious mosquitoes resident on that patch.

We will use $I_M$, $I_N$ to denote the compartments of blood-infected moving and non-moving individuals, respectively. The liver-infected compartments $L_M$, $L_N$ -- containing individuals with latent hypnozoites, but no active blood infections -- are defined similarly. All other individuals are placed into the susceptible $S_M$ compartment (if moving) or the $S_N$ compartment (if non-moving). In particular, $M = I_M + L_M + S_M$ and $N = I_N + L_N + S_N$. The proportion of infectious mosquitoes on Patch $i$ will be given by $I_{mi}$, $i \in \{1, 2\}$.

\subsection{Within-host model}

\label{withinhost}

We begin by capturing within-host dynamics for an arbitrary non-moving individual. Here, we apply the analysis and resulting probability generating function (PGF) derived in \cite{Mehra}, which describes the stochastic primary infection and hypnozoite-related processes in the case of an individual in a homogeneous environment, as a function of the FORI associated with the environment.

Considering a single hypnozoite established in the host immediately after an infective bite at time $t = 0$, we let the initial \textit{establishment} state of the hypnozoite in a host hepatocyte be $H$. From Equation (13) in \cite{Mehra}, the probability $p_H(t)$ that a single short-latency hypnozoite established at time $0$ is in state $H$ at time $t$ satisfies
\be\label{hypnozoite1} \begin{aligned} & p_H(t) = e^{-\pa{\alpha + \mu}t}.  \end{aligned} \ee
An established hypnozoite may activate, transitioning to a state $A$ at rate $\alpha$ and causing a relapse infection that is cleared at rate $\gamma$ (at which point the hypnozoite enters a cleared state $C$). From Equations (14)-(15) in \cite{Mehra}, the probabilities $p_A(t), p_C(t)$ that a hypnozoite established at time $0$ is activated or cleared, respectively, at time $t$ are given by
\be\label{hypnozoite2} \begin{aligned}& p_{A}(t) = \frac{\alpha}{\alpha + \mu - \gamma}\pa{e^{-\gamma t} - e^{-(\alpha + \mu) t}}, \\ & p_C(t) = \frac{\alpha}{\alpha + \mu}\pa{1 - e^{-(\alpha + \mu)t}} - \frac{\alpha}{\alpha + \mu - \gamma} \pa{e^{-\gamma t} - e^{-(\alpha + \mu)t}}.  \end{aligned} \ee
Alternatively, an established hypnozoite may die, transitioning to a state $D$ at rate $\mu$. By Equation (16) in \cite{Mehra}, the probability $p_D(t)$ that a hypnozoite established at time $0$ is dead at time $t$ is given by
\be\label{hypnozoite3} \begin{aligned}&  p_D(t) = \frac{\mu}{\alpha + \mu} \pa{1 - e^{-(\alpha + \mu)t}}.  \end{aligned} \ee
In addition to defining states for each inoculated hypnozoite, we refer to an active primary infection as being in state $P$. Similarly to relapses, primary infections are cleared at rate $\gamma$, transitioning to a state $PC$. This chain of transitions is illustrated in Figure 2.
\begin{figure}
     \centering
         \includegraphics[width= 0.8\textwidth]{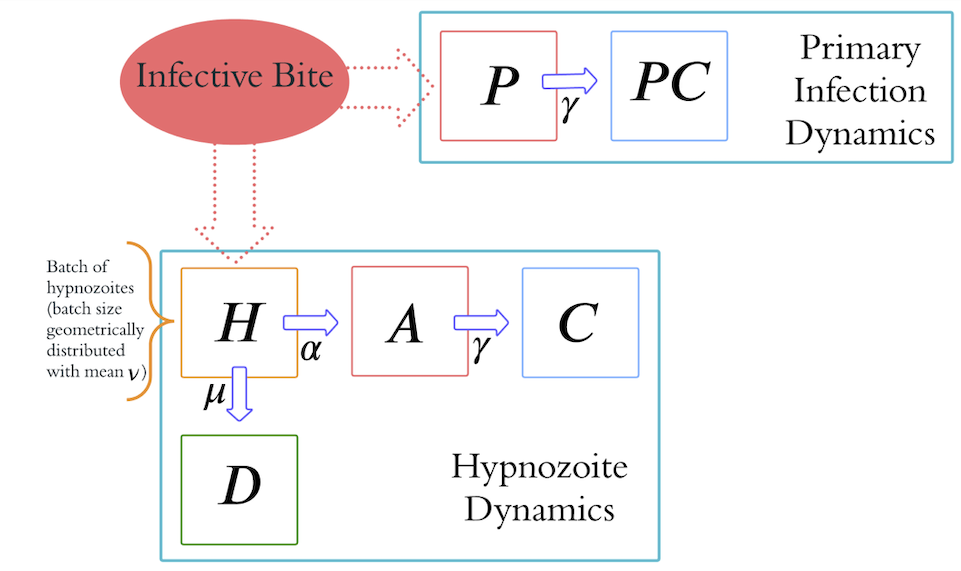}
     \caption{Schematic representation of primary infection and within-host dynamics for a single individual, based on the analysis in \cite{Mehra}. An infective bite (top left) establishes a primary infection (in within-host compartment $P$) and a batch of hypnozoites (initially in latent within-host compartment $H$), which activate at rate $\alpha$ and die at rate $\mu$, moving to within-host compartments $A$, $D$ respectively. Both primary infections and secondary infections caused by activated hypnozoites clear at rate $\gamma$, moving to compartments $PC$, $C$, respectively.}         \label{Schematic2}
\end{figure}

We assume that the time-dependent random variables representing the state of each hypnozoite dormant within the host are independent and identically-distributed (e.g., the probabilities in \eqref{hypnozoite1}, \eqref{hypnozoite2},\eqref{hypnozoite3} apply to each hypnozoite contained in the host). When a vector infects the individual with \textit{P. vivax} parasites, the individual receives both a primary infection and an inoculation of $h$ hypnozoites (with the probability mass function for $h$ being geometrically-distributed, and with $h$ having expected value $\nu$ and associated state space $[0, \infty) \cup \mathbb{Z}$), see Figure~\ref{Schematic2}.

\subsection{Probability-generating functions for within-host dynamics within moving and non-moving individuals}

The rate of infective bites on Patches $1$ and $2$ on the interval $[0, t]$  are modeled by non-homogeneous Poisson processes with rates $\lambda_1(t)$, $\lambda_2(t)$ respectively. We assume that hypnozoite inoculation is instantaneous. 

We denote the number of hypnozoites in the established, activated, cleared, and dead states by $\mathcal{N}_H(t), \mathcal{N}_A(t), \mathcal{N}_C(t), \mathcal{N}_D(t)$, respectively, and also denote the number of primary infections and cleared infections by $\mathcal{N}_P(t)$ and $\mathcal{N}_{PC}(t)$, respectively. 

Letting $F' = \{H, A, C, D, P, PC\}$ and $F =  \{H, A, C, D\}$, by \cite{Mehra}, the PGF yielding within-host dynamics for individuals in $\mathcal{N}$ is given by
\begin{equation}\label{eq:pgforiginal}
\begin{aligned} 
& G_{N}(z_H, z_A, z_C, z_D, z_P, z_{PC}) :=   \\ & \mathbb{E} [\prod_{f \in F'} z_{f}^{\mathcal{N}_f(t)}] = \exp \int_0^t  \pa{\lambda_{1}(\tau)\pa{\frac{\pa{z_{P}e^{-\gamma(t - \tau)} + (1 - e^{-\gamma(t - \tau)})z_{PC}}}{1 + \nu\pa{1 - \sum_{f \in F} z_{f} \cdot p_{f}(t - \tau)}} \hspace{3pt}} - 1} d\tau.
\end{aligned}
\end{equation}

We now extend the above PGF for an individual in the moving group $\mathcal{M}$. To the latter individual, we associate a discrete piecewise-continuous \textit{movement function} $\mathcal{B}_j(t):= \mathbb{R}^+ \to \{0, 1\}$, equal to $1$ when the individual $j$ is on Patch 1 (their home village) and to $0$ otherwise. The derivation of the PGF characterizing within-host dynamics for the moving individual is unaltered up to a change in the force-of-reinfection experienced by the individual. The within-host PGF for an individual in $\mathcal{M}$ with movement function $\mathcal{B}_j(t)$ thus becomes

\begin{equation}\label{eq:pgfmoving}
\begin{aligned} 
& G^{(j)}_{M}(z_H, z_A, z_C, z_D, z_P, z_{PC}) :=   \\ & \mathbb{E} [\prod_{f \in F'} z_{f}^{\mathcal{N}_f(t)}] = \exp \int_0^t  \Biggr[\mathcal{B}_j(\tau)\lambda_{1}(\tau) + \pa{1 - \mathcal{B}_j(\tau)}\lambda_{2}(\tau)\Biggr] \Biggr[\frac{\pa{z_{P}e^{-\gamma(t - \tau)} + (1 - e^{-\gamma(t - \tau)})z_{PC}}}{1 + \nu\pa{1 - \sum_{f \in F} z_{f} \cdot p_{f}(t - \tau)}} - 1\Biggr] \hspace{3pt} d\tau.
\end{aligned}
\end{equation}

Each function $\lambda_i(t), i \in \{1, 2\}$ depends both on the proportion of infectious mosquitoes (relative to the entire mosquito population) and on the mosquito-to-human ratio on this patch. We write 
\begin{equation}\label{lambdaeq} \begin{aligned} \lambda_1(t) = ab\ \cdot \ \frac{I_{m1}(t)}{Q_1} \ \cdot  \ \frac{Q_1}{\pa{N + (u + v \cos \omega t) \, M}} = \frac{ab\,I_{m1}(t)}{\pa{N + (u + v \cos \omega t) \, M}}; \\ \lambda_2(t) = ab\ \cdot \ \frac{I_{m2}(t)}{Q_2} \ \cdot  \ \frac{Q_2}{M \pa{(1-u) - v \cos \omega t}} = \frac{ab\,I_{m2}(t)}{M \pa{(1-u) - v \cos \omega t}} \end{aligned} \end{equation} by \eqref{A1A2}.
Here, the FORI on Patch $i$, $i \in \{1, 2\}$, is the product of $ab$ -- which represents the mean mosquito bite rate $a$ (in days$^{-1}$) across both patches multiplied by the mean probability of mosquito-to-human transmission $b$ -- with the proportion of infectious mosquitoes on Patch $i$, and the mosquito-to-human ratio on Patch $i$. 
Using the shorthands
\[f(t) := \frac{- e^{- \gamma t} - \nu p_{A}(t)}{1 + \nu p_{A}(t)}\] 
and
\[g(t):=  \frac{- e^{- \gamma t} - \nu p_{A}(t) - \nu p_H(t)}{1 + \nu p_{A}(t) + \nu p_H(t)},\] we apply \eqref{eq:pgforiginal} and \eqref{eq:pgfmoving} to characterize the probabilities that individual humans are liver- or blood-infected. We note here that, by Equation (51) in \cite{Mehra}, $\abs{f(t)}$ captures the probability that an infective bite at time $0$ leads to a blood-stage infection at time $t$. Similarly, $\abs{g(t)}$ captures the probability that an infective bite at time $0$ leads to either a blood-stage or a liver-stage infection at time $t$.

Following \cite{Anwar}, we note that the probability of an individual in $\mathcal{N}$ to be blood-infected (i.e., in the compartment $I_N$) is given by \be\label{probbloodn} \begin{aligned} 1 - G_N(z_H = 1, z_A = 0, z_C = 1, z_D = 1, z_P = 0, z_{PC} = 1) = 1 - \exp \pa{\int_{0}^{t} \lambda_1(\tau) f(t - \tau) \, d\tau}. \end{aligned} \ee Moreover, the probability that an individual in $\mathcal{M}$ characterized by movement function $\mathcal{B}_j(t)$ is blood-infected (i.e., in $I_M$) can be written as 
\be\label{probbloodm} \begin{aligned} G^{(j)}_M(z_H = 1, z_A = 0, z_C = 1, z_D = 1, z_P = 0, z_{PC} = 1) = \\ 1 - \exp \pa{\int_{0}^{t} \pa{\mathcal{B}_j(\tau) \lambda_1(\tau) + (1-\mathcal{B}_j(\tau))\lambda_2(\tau)} f(t - \tau) \, d\tau}. \end{aligned} \ee  
Similarly, the probability that an individual in $\mathcal{N}$ belongs to $L_N$ (the liver-infected stationary population) at time $t$ is given by 
\be\label{problivern} \begin{aligned} G_N(z_H = 1, z_A = 0, z_C = 1, z_D = 1, z_P = 0, z_{PC} = 1) \\ - \, G_N(z_H = 0, z_A = 0, z_C = 1, z_D = 1, z_P = 0, z_{PC} = 1) = \\ \exp\pa{\int_{0}^{t}\lambda_1(t)\, f(t - \tau) \, d \tau} -  \exp\pa{\int_{0}^{t}\lambda_1(t)\, g(t - \tau) \, d \tau},\end{aligned} \ee 
while the probability that an individual in $\mathcal{M}$ migrating with the same movement function as above belongs to $L_M$ (liver-infected migratory population) 
 at time $t$ is
\be\label{probliverm} \begin{aligned} G^{(j)}_M(z_H = 1, z_A = 0, z_C = 1, z_D = 1, z_P = 0, z_{PC} = 1) \\ - \, G^{(j)}_M(z_H = 0, z_A = 0, z_C = 1, z_D = 1, z_P = 0, z_{PC} = 1) =  \\ \exp\pa{\int_{0}^{t}\pa{\mathcal{B}_j(\tau)\lambda_1(\tau) + (1-\mathcal{B}_j(\tau))\lambda_2(\tau)}\, f(t - \tau) \, d \tau} - \\ \exp\pa{\int_{0}^{t}\pa{\mathcal{B}_j(\tau)\lambda_1(\tau) + (1-\mathcal{B}_j(\tau))\lambda_2(\tau)}\, g(t - \tau) \, d \tau}.\end{aligned} \ee

We note that, in order for differential equations to be used in describing population dynamics, the number of humans in the population must be modeled as a continuum. 
In other words, there must exist a bijection  between the space of moving humans and the closed interval $[0, M]$. We use this in Section \ref{eq:deriv} to tie the movement functions of each individuals $i \in \mathcal{M}$ to the functions $A_1(t), A_2(t)$. 

Figure \ref{comparplot} illustrates the probabilities \eqref{probbloodm} and \eqref{probliverm} for two moving individuals migrating with different periods relative to the population's mean period of movement ($2\pi/\omega$). We note that the slow-moving individual exhibits far greater amplitude of oscillations in infection probability than the rapidly-moving individual, and that both individuals' probabilities of becoming liver- and blood-infected are periodic functions with multiple frequencies. Oscillations in infection probability vary across time, with variations coinciding with changes in the rate of mean population migration. 

In Section \ref{population-level}, we will use the probabilities derived above (particularly \eqref{probbloodn} and \eqref{probbloodm}) to formulate population-level equations governing the time evolution of infected humans.

\begin{figure}[h!]
\begin{center}
\includegraphics[width=\textwidth]{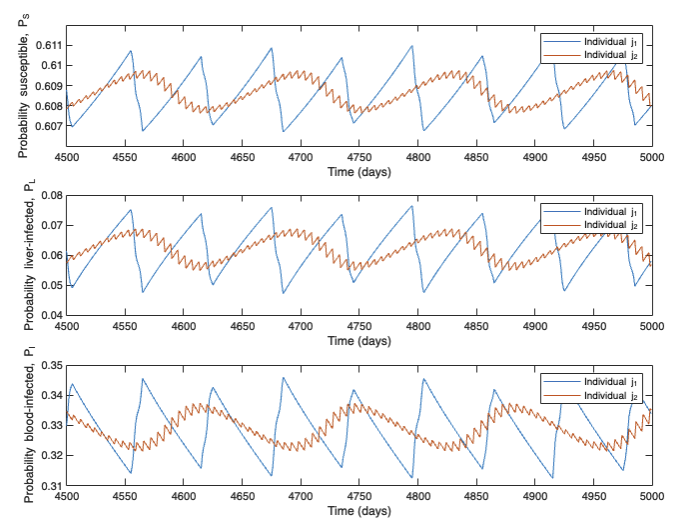}
\end{center}
\caption{ Comparison of probabilities for a slow-moving ($j_{1} \in \mathcal{M}$) and fast-moving ($j_{2} \in \mathcal{M}$) individual to be susceptible ($P_S$), liver-infected ($P_L$), and blood-infected ($P_I$), respectively, at time $t \in [4500, 5000]$. We assume the parameter values to be those in Table \ref{parameters}, so that the mean period of population movement is $2\pi/\omega = 2\pi$ days. The fast-moving individual (brown lines) switches patch at a frequency of $1 \ \textnormal{day}^{-1}$, above the population mean frequency; the slow-moving individual (blue lines) travels at a frequency of $\frac{1}{60} \, \textnormal{days}^{-1}$, below the population mean frequency. Probabilities are derived from Equations \eqref{lambdaeq}, \eqref{probbloodm}, and \eqref{probliverm}. }\label{comparplot}
\end{figure} 

\begin{table}[h!]\label{table1}
\centering\begin{tabular}{| m{1em} || m{16em}  m{2.5cm} m{3.3cm} | } \hline
  & \centering Definition & \centering Value & \hspace{10pt} Source \\ [0.5ex] 
 \hline\hline
  $a$ & 
  Mean daily mosquito bite rate & \centering 0.21 days$^{-1}$ & \hspace{25 pt}  \cite{garrett1964human} \\ \hline
  $b$ & 
  Probability of mosquito to human transmission & \centering 0.5* & \hspace{25 pt}  \cite{smith2010quantitative} \\ \hline
  $c$ & 
  Probability of human to mosquito transmission & \centering 0.23* & \hspace{25 pt}  \cite{bharti2006experimental} \\ \hline
  $g_1$ & 
  Patch $1$ mosquito demography rate & \centering 0.1 days$^{-1}$ & \hspace{25 pt}  \cite{gething2011modelling} \\ \hline
  $g_2$ & 
  Patch $2$ mosquito demography rate & \centering 0.08 days$^{-1}$ & \hspace{25 pt} \cite{tan2008bionomics} \\ \hline
  $n$ & 
  Mosquito sporogyny rate & \centering $1/12$ days$^{-1}$ & \hspace{25 pt}  \cite{gething2011modelling} \\ \hline
  $Q_1$ & 
  Proportion of mosquitoes on Patch $1$ relative to number of humans & \centering 0.116* & \hspace{25pt}\cite{overgaard2003effect,white2014modelling} \\ \hline
  $Q_2$ & 
  Proportion of mosquitoes on Patch $2$ relative to number of humans & \centering 0.464* & \hspace{25pt}\cite{overgaard2003effect,white2014modelling} \\ \hline
  $u$ & 
  $A(0)$, where $A(t) = u + v \cos \omega t$ is the density of moving humans on Patch $1$ at time $t$ & \centering 0.5 days$^{-1}$& \hspace{9pt}$\substack{\text{\normalsize Arbitrary,} \\ \text{ \normalsize later varied}}$\\ \hline
  $v$ & 
  $A(0) - A(\pi/2)$ & \centering -0.25 days$^{-1}$& \hspace{5pt} $\substack{\text{\normalsize Arbitrary,} \\ \text{ \normalsize later varied}}$\\ \hline
  $\omega$ & 
  $(2\pi)\, \times$ the frequency of human movement & \centering 1 day$^{-1}$ & \hspace{-17pt} $\substack{\text{\normalsize Chosen due to the} \\ \text{ \normalsize travel patterns } \\ \text{ \normalsize mentioned in \cite{bannister2019forest} }}$ \\ \hline
   $\alpha$ & 
  Hypnozoite activation rate & \centering 1/332 days$^{-1}$ &  \hspace{25 pt}  \cite{white2014modelling} \\ \hline
   $\mu$ & 
  Hypnozoite death rate & \centering 1/425 days$^{-1}$ & \hspace{25 pt} \cite{white2014modelling} \\ \hline
   $\gamma$ & 
  Blood-infection clearance rate & \centering 1/60 days$^{-1}$ & \hspace{25 pt}  \cite{collins2003retrospective} \\ \hline
   $\nu$ & Mean number of hypnozoites established in a bite
   & \centering $5$* & \hspace{25 pt}  \cite{white2016variation} \\ \hline
\end{tabular}
\caption{List of baseline values for model parameters, used when generating numerical solutions. Most parameter values are retrieved from \cite{Anwar}; original sources are listed in the fourth column. Asterisks indicate a dimensionless parameter. Given the model's representation of village-forest patch geography and focus on \textit{P. vivax} dynamics, parameter values are derived for Southeast Asia and the Americas.}\label{parameters}
\end{table}

\subsection{Population-level model}\label{population-level}

 We now turn from the individual-level scale to the population-level one. In our model, we do not explicitly represent population-level interactions between the six human compartments defined in Section \ref{withinhost}. Rather, we use the transitions between compartments embedded inside the within-host model to characterize all compartment densities as functions of the two forces of reinfection \eqref{lambdaeq}.

 \subsubsection{Population-level densities of infectious moving and non-moving individuals}
 \label{eq:deriv}

 Since the non-moving population is homogeneous with respect to hypnozoite and infection dynamics and the number of non-moving humans is assumed to be arbitrarily large, the density of $I_N(t)$ approaches the product of $N$ and the probability that a single non-moving individual is infectious (see, e.g., the hybrid modeling strategies in \cite{mehra2023superinfection, nasell2013hybrid}). This probability is given by \eqref{probbloodn} and hence $I_N(t)$ is equal to
 \[I_N(t) = N - N \exp\pa{\int_{0}^{t}\lambda_1(t)\, f(t - \tau) \, d \tau}.\] 

 Similarly, we can closely approximate the density of infectious moving individuals in the population by the expected size of $I_M(t)$.  
 To calculate $I_M(t)$, we let there be a bijection mapping each 
 moving person $j \in \mathcal{M}$ to a periodic, deterministic \textit{movement function} $\mathcal{B}_j(t)$, equal to $1$ if the corresponding individual is on Patch $1$ and to $0$ otherwise.  
 We note that the probability of an individual with movement function $\mathcal{B}_j(t)$ to be in $I_M$ at time $t$ is, by \eqref{probbloodm} in Section \ref{withinhost}, 
 \[p_j(t) = 1 - \exp\pa{\int_{0}^{t}  \pa{\underbrace{\mathcal{B}_j(\tau)}_{\substack{\text{Patch 1} \\ \text{indicator} \\ \text{function}}} \underbrace{\lambda_1(\tau)}_{\substack{\text{Patch 1} \\ \text{FORI}}} + \underbrace{(1 - \mathcal{B}_j(\tau))}_{\substack{\text{Patch 2} \\ \text{indicator} \\ \text{function}}} \underbrace{\lambda_2(\tau)}_{\substack{\text{Patch 2} \\ \text{FORI}}}}f(t - \tau)\ d\tau}.\]

 Given $i \in \mathcal{M}$, we assume that, for all $j \in \mathcal{M}$, \be\label{mjt}\begin{aligned}\mathcal{B}_j(t) = \mathcal{B}_i(t + \tau_{j})\end{aligned}\ee for some $\abs{\tau_{j}} << 1 \in \mathbb{R}$ and all $t \in [0, \infty)$, i.e., all movement functions coincide up to arbitrarily-small time-shifts. For simplicity and ease of analysis, we will replace each $\mathcal{B}_j(t)$ with its first-order approximation -- i.e., the first two terms in the Fourier series for $\mathcal{B}_j(t)$. 
 \begin{lem}\label{lem1} The first-order approximation $B_j(t)$ of $\mathcal{B}_j(t)$ satisfies \[\lim_{\sum_{k \in M} \abs{\tau_{k}} \to \, 0}B_j(t) = u + v \cos(\omega t),\] where $u, v, \omega$ are defined in \eqref{A1A2}. \end{lem}

Given Lemma \ref{lem1}, we can make the approximation \[\mathcal{B}_j(t) \approx B_j(t) := u + v \cos\pa{\omega t},\] from where it follows that 
\begin{equation}\label{densityofmovinginf} \begin{aligned} I_M(t) \approx M - M \exp\pa{\int_{0}^{t} \pa{A_1(\tau) \lambda_1(\tau) + A_2(\tau) \lambda_2(\tau)}f(t - \tau)\ d\tau}. \end{aligned} \ee

For completeness, we also formulate expressions for $L_N(t)$ and $L_M(t)$. Via a similar argument to the above, we obtain that the density of liver-infected non-moving individuals in the population is 
\be\label{livdensityn} \begin{aligned} L_N(t) = \underbrace{N \exp\pa{\int_{0}^{t}\lambda_1(t)\, f(t - \tau) \, d \tau}}_{\substack{\text{density of non-moving individuals} \\ \text{without blood infections}}} -  \underbrace{N \exp\pa{\int_{0}^{t}\lambda_1(t)\, g(t - \tau) \, d \tau}}_{\substack{\text{density of non-moving individuals} \\ \text{without liver or blood infections}}},\end{aligned} \ee and
the density of liver-infected moving individuals is 
\begin{equation}\label{livdensity} \begin{aligned} L_M(t) \approx \underbrace{M \exp\pa{\int_{0}^{t}\pa{\underbrace{A_1(\tau)}_{\substack{\text{Patch $1$} \\ \text{density}}} \lambda_{1}(\tau) + \underbrace{A_2(\tau)}_{\substack{\text{Patch $2$} \\ \text{density}}} \lambda_{2}(\tau)}\, f(t - \tau) \, d \tau}}_{\substack{\text{density of moving individuals without blood infections}}} -  \\ \underbrace{M \exp\pa{\int_{0}^{t}\pa{\underbrace{A_1(\tau)}_{\substack{\text{Patch $1$} \\ \text{density}}} \lambda_{1}(\tau) + \underbrace{A_2(\tau)}_{\substack{\text{Patch $2$} \\ \text{density}}} \lambda_{2}(\tau)}\, g(t - \tau) \, d \tau}}_{\text{density of moving individuals without liver or blood infections}}.\end{aligned} \end{equation}

\begin{proof}[Proof of Lemma \ref{lem1}]
    We first note that the density of individuals on Patch 1 can be written as
    \[A_1(t) = \frac{1}{M} \int_j \mathcal{B}_j(t) \ dj.\] Since all $\mathcal{B}_j$ have the same period, $\mathcal{B}_j(t) = \mathcal{B}_j(t + 2\pi/\omega)$ for all $j \in M$. Fixing $j \in \mathcal{M}$, we compute that the constant term of the Fourier series for $\mathcal{B}_j$ is 
    \be \begin{aligned} \frac{\omega}{2\pi}\int_0^{2\pi/\omega} \mathcal{B}_j(\tau) \ d \tau = \frac{1}{M} \int_j \pa{\frac{\omega}{2\pi}\int_0^{2\pi/\omega} \mathcal{B}_j(\tau) \ d \tau}  d j = \frac{\omega}{2\pi} \int_0^{2\pi/\omega} A_1(\tau) \ d \tau = u.\end{aligned} \ee 

    Similarly, upon decomposing all $\mathcal{B}_j$ into Fourier series, we note that, by \eqref{mjt}, the second terms of all such Fourier series must be of the form $x\cos(\omega t + \tau_j) + y\sin(\omega t + \tau_j)$. Moreover, $\lim_{\tau_j \to 0} x = v, \lim_{\tau_j \to 0} y = 0$ since $\frac{1}{M} \int_j x\cos(\omega t + \tau_j) + y\sin(\omega t + \tau_j) dj = v \cos (\omega t)$. The statement of the lemma follows. 
\end{proof}

\subsubsection{Derivation of population-level model with coupled human-mosquito dynamics}\label{deriv:population}To represent mosquito dynamics on both patches, we construct a nonlinear SEI-type model based on \cite{Anwar}. We divide the sets $Q_1$, $Q_2$ into three sub-populations each, denoted by $S_{mi}, E_{mi}, I_{mi}$, $i \in {1, 2}$, and representing susceptible, exposed, and infectious mosquitoes, respectively (Figure \ref{Schematic(3)}). Here $g_i$ denotes the (constant as opposed to seasonal) birth and death rates for mosquitoes on Patch $i$, with $g_1 > g_2$. Mosquito sporogony occurs at rate $n$. The force of infection $\Phi_i(t)$ for mosquitoes on Patch $i$, $i \in {1, 2}$, is the product of the mean mosquito bite rate, $a$, the probability of human-to-mosquito transmission, $c$, and the proportion of infectious humans on Patch $i$ (relative to the total number of humans on Patch $i$). In particular,
\begin{equation}\begin{aligned}& \Phi_1(t) = \frac{ac \, \pa{I_{N} + (u - v \cos \omega t) \, I_{M}}}{N + (u + v \cos \omega t) \, M}; \\ & \Phi_2(t) = \frac{ ac\,I_M}{M}.
\end{aligned}\end{equation}

\begin{figure}
     \centering    \includegraphics[width=0.6\textwidth]{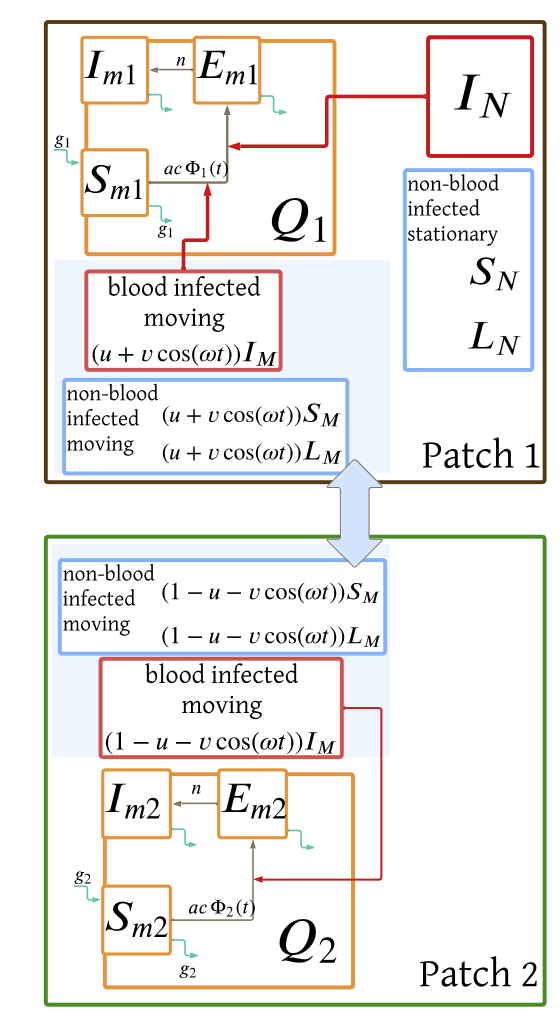}
     \caption{Schematic of the two-patch population-level model. Within-host processes (including  recovery, primary infection, or hypnozoite activation mechanisms) is omitted to best display population-level dynamics. The red arrows indicate infection mechanisms, while the large blue arrow represents the migration of moving individuals between Patch 1 (the upper large square) and Patch 2 (the lower large square). Within each patch, the group of moving individuals (including blood-infected, liver-infected, and susceptible travelers) is indicated with a lightly-shaded blue rectangle.}\label{Schematic(3)}
\end{figure} 

To determine $I_{m1}(t)$ and $I_{m2}(t)$ -- functions which then allow us to describe the infection and hypnozoite dynamics of the entire moving and non-moving population -- we solely require knowledge of the density of infected non-moving humans, $I_N(t)$, and the density of infected moving humans, $I_M(t)$. As a result, to complete our ``minimal" population-level model, we only need equations for $I_N(t)$ and $I_M(t)$. These equations follow directly from Section \ref{eq:deriv}. Our population-level model can thus be described by the following system of differential equations (with embedded  Volterra convolution equations of the first kind):

\begin{equation}\label{eq:humsimplify}
\begin{aligned}
    &\frac{d S_{m1}}{dt} = g_{1}Q_1 - \frac{ac \, \pa{I_{N} + (u + v \cos \omega t) \, I_{M}}}{N + (u + v \cos \omega t) \, M}\, S_{m1} - g_{1} S_{m1}; \\
    & \frac{d E_{m1}}{dt} = \frac{ac \, \pa{I_{N} + (u + v \cos \omega t) \, I_{M}}}{N + (u + v \cos \omega t) \, M}\, S_{m1} - (g_{1} + n)E_{m1}; \\
    & \frac{d I_{m1}}{dt} = n E_{m1} - g_{1} I_{m1}; \\
    &\frac{d S_{m2}}{dt} = g_{2}Q_2 - \frac{ ac\,I_M}{M} S_{m2} - g_{2} S_{m2}; \\
    & \frac{d E_{m2}}{dt} = \frac{ ac\,I_M}{M} S_{m2} - (g_{2} + n)E_{m2}; \\
    & \frac{d I_{m2}}{dt} = n E_{m2} - g_{2} I_{m2};\\
    & I_{M} = M  - M \exp\pa{ab\,\int_{0}^{t}  \pa{\frac{I_{m1}(\tau)(u + v \cos \omega \tau)}{N + (u + v \cos \omega t)M}  \, + \frac{I_{m2}(\tau)((1 - u) - v\cos \omega \tau)}{((1 - u) - v\cos \omega \tau)M}} f(t - \tau)\, d\tau}; \\
    & I_{N} =  N - N \exp\pa{ab\,\int_{0}^{t}\frac{I_{m1} (\tau)}{N +  (u + v\cos \omega \tau) \, M}\, f(t - \tau) \, d \tau} \\
  \end{aligned}
\end{equation}
with initial conditions 
\be \begin{aligned} & S_{m1}(0) = S^{(0)}_{m1}; \hspace{3pt} E_{m1}(0) = E^{(0)}_{m1}; \hspace{3pt}I_{m1}(0) = I^{(0)}_{m1}; \\ &   \hspace{3pt}S_{m2}(0) = S^{(0)}_{m2}; \hspace{3pt}E_{m2}(0) = E^{(0)}_{m2}; 
\hspace{3pt}I_{m2}(0) = I^{(0)}_{m2}; \hspace{3pt} I_M(0) = 0; \hspace{3pt} I_N(0) = 0. \end{aligned} \ee

We note that this minimal model system does not incorporate the equation for the densities of individuals with liver-only infections. As described above, the solutions for $L_N(t), L_M(t)$ (unlike those for $I_N(t), I_M(t)$) do not directly determine the distribution of mosquitoes across compartments in the population-level model. Moreover, the equations for $L_N(t), L_M(t)$ are uncoupled from those for $I_N(t), I_M(t)$ via the analysis in Section \ref{withinhost}, meaning that \eqref{livdensityn} and \eqref{livdensity} can be excluded from the model in \eqref{eq:humsimplify}.
We return to the implications of the solutions for $I_{m1}(t), I_{m2}(t)$ on $L_N(t), L_M(t)$ in the sections that follow.

\subsubsection{Limiting prevalence assuming constant forces of re-infection}\label{constantFORIsection}

We now consider a situation in which the forces of re-infection (FORIs) on patches $1$ and $2$ -- $\lambda_1, \lambda_2$, respectively -- are \textit{constant}. In this case, a single equation may be derived for $I_M$, which will yield insight for the impact of movement on transmission. 

As before, we let the expected proportion of people on patches $1$ and $2$ at time $t$ be of the form $u + v  \cos(\omega t), 1 - u - v \cos(\omega t)$, respectively. 

From \eqref{densityofmovinginf}, the density of infectious moving individuals is  

\begin{equation}\label{imt} \begin{aligned} I_M(t) \approx M - M \exp \pa{\lambda_1 \int_0^t \pa{u + v\cos(\omega t)} \, f(t - \tau) \, d\tau  + \lambda_2 \int_0^t \pa{1 - u - v\cos(\omega t)} \, f(t - \tau) \, d\tau}. \end{aligned} \end{equation}

We define $I^{*}_M(t)$ as the periodic function satisfying $\lim_{t \to \infty} \abs{I_M(t) - I^{*}_M(t)} = 0$. We seek $I^{*}_M(t)$ in terms of the parameters, including $u, v$, and $\omega$. We note that

\begin{equation}\label{toforisub} \begin{aligned} I_M(t) = M - M \exp \pa{(\lambda_1 u + \lambda_2 (1 - u)) \int_{0}^{t} f(\tau) \ d\tau} \exp \pa{(v\lambda_1 -v\lambda_2) \int_0^t \cos(\omega t) f(t - \tau) \, d\tau},\end{aligned} \end{equation}
where \begin{equation} \begin{aligned} \int_{0}^t \cos(\omega \tau) f(t - \tau) \, d\tau =  \frac{1}{2} e^{\omega \textbf{i} t} \int_{0}^t \pa{e^{-\omega \textbf{i} \tau}}  \ f(\tau)\, d\tau + \frac{1}{2} e^{-\omega \textbf{i} t} \int_{0}^t \pa{e^{\omega \textbf{i} \tau}} \ f(\tau)\, d\tau.\end{aligned} \end{equation} 

We observe that $I^{*}_M(t)$ is a periodic function with period $\frac{2\pi}{\omega}$. This follows from the bound

\[\lim_{t \to \infty} \int_0^t \frac{-e^{- \gamma \tau} - \nu p_{A}(\tau)}{1 + \nu p_A\pa{\frac{ \ln\pa{\gamma/(\alpha + \mu)}}{-\alpha - \mu + \gamma}}} \ d\tau > \lim_{t \to \infty} \int_0^t f(\tau) \ d\tau >  \lim_{t \to \infty} \int_0^t \pa{-e^{- \gamma \tau} - \nu p_{A}(\tau)}  \ d\tau\] 

which implies that the improper integrals 

\be\label{C0} \begin{aligned} C_0 = \lim_{t \to \infty} \int_0^t f(\tau) \ d\tau; \hspace{5pt} C_1 := \lim_{t \to \infty} \int_0^t e^{-\omega \textbf{i} \tau} f(\tau) \ d\tau; \hspace{5pt} C_{-1} := \lim_{t \to \infty} \int_0^t e^{\omega \textbf{i} \tau} f(\tau) \ d\tau,\end{aligned}\ee converge to constants in $\mathbb{C}$.

Fixing $t >> 1$, from \eqref{tosub}, we see that $I^{*}_M(t)$ is a monotonically-decreasing function of $u$ (where we use the fact that $\lambda_1 < \lambda_2$). Moreover, since  
 $\max_{t} I^{*}_M(t) - \min_{t} I^{*}_M(t)$ is proportional to the value of 
\[\exp \pa{v\lambda_1 -v\lambda_2} \pa{ \max_{t} \exp \pa{ \int_0^t \cos(\omega t) f(t - \tau) \, d\tau} - \min_{t} \exp \pa{\int_0^t \cos(\omega t) f(t - \tau) \, d\tau}},\] 
$\max_{t} I^{*}_M(t) - \min_{t} I^{*}_M(t)$ is a monotonically-decreasing function of $\abs{v}$ and a monotonically-increasing function of $\omega$ (an equivalent statement will be shown for the non-constant FORI system \eqref{eq:humsimplify} in Appendix A).

For sufficiently-large $\omega$ and sufficiently-small $\abs{v}$, we observe that $I_M(t)$ can be closely decomposed by a constant term and two oscillating terms with frequency $1/\omega$. This is relevant to our choice of ansatz when seeking solutions for the full population-level model with coupled human-vector dynamics.

\subsection{Numerical solution of multiscale model}\label{numerical_solution}

In order to simulate the model in \eqref{eq:humsimplify}, we implemented a solver for integro-differential equationswith step size $t_{\textnormal{step}}$ (see Algorithm \ref{alg:1}). For each value of $t \in {k \cdot t_{\textnormal{step}}}, k \in \mathbb{Z}, k \in [0, 5000/t_{\textnormal{step}}]$, the solver generated values for $I_M(t)$ and $I_N(t)$ using the history of the FORI. The solver then used the newest values for $I_M(t)$ and $I_N(t)$ to update the values of $S_{m1}$, $S_{m2}$, $E_{m1}$, $E_{m2}$, $I_{m1}$, $I_{m2}$ via the fourth-order Runge Kutta. The initial conditions used were $S_{m1} = 0.116, S_{m2} = 0.406, E_{m2} = 0.058, I_{m2} = E_{m1} = I_{m1} = I_N = I_M = 0$. 
\begin{enumerate}
\begin{algorithm}
\caption{}\label{alg:1}
\item Set $t = 0$ and define values for $t_{\textnormal{end}}, t_{\textnormal{step}}$
\item 
\begin{algorithmic}
\While{$t < t_{\textnormal{end}}$}

Compute $I_N(t), I_M(t)$ using a Darboux sum with $I_{m1}(\tau), I_{m2}(\tau)$, $\tau \in [0, t]$
\vspace{10pt}

Update $S_{mi}(t + t_{\textnormal{step}}), E_{mi}(t + t_{\textnormal{step}}), I_{mi}(t + t_{\textnormal{step}})$, $i \in {1, 2}$ using $I_N(t), I_M(t)$ and fourth-order Runge-Kutta (1000 steps) on the time interval $(t, t + t_{\textnormal{step}})$
\vspace{10pt}

Set $t = t + t_{\textnormal{step}}$
\EndWhile
\end{algorithmic}
\end{algorithm}
\end{enumerate}

\section{Results} \label{results}
\subsection{Existence and location of the limiting-state solutions of System \eqref{eq:humsimplify}}\label{approx}

Given non-zero $v$, the system \eqref{eq:humsimplify} of integro-differential equations is non-autonomous. Setting the left hand sides of the system to zero indicates that \eqref{eq:humsimplify} possesses only one equilibrium solution, namely the disease-free equilibrium (DFE), satisfying $E_{m1} = E_{m2} = I_{m1} = I_{m2} = I_M = I_N = 0$, $S_{m1} = Q_1$, $S_{m2} = Q_2$. Moreover, one can observe from the right hand sides of the equations in \eqref{eq:humsimplify} that the DFE is the only fixed-point limit set associated to an orbit of the dynamical system defined by \eqref{eq:humsimplify}.

Numerical experiments (see Section \ref{numresults}) indicate that the endemic solutions to \eqref{eq:humsimplify} are asymptotically non-constant periodic, which is further supported by our analysis in the case when the FORI is constant (see Section \ref{constantFORIsection}). 
In this section, we make the assumption that the biologically-realistic solutions to \eqref{eq:humsimplify} are limit-periodic, and analyze the location of endemic solutions when they exist.  

Accordingly, we seek solutions to \eqref{eq:humsimplify} that can be decomposed as
\begin{equation}\label{coefffull}
\begin{aligned}
& S_{mi}(t) = \sum_{n = -\infty}^{\infty} u^{(i)}_n(t) \exp\pa{\omega\, \textbf{i}\, n \, t}, \hspace{5pt} i \in \{1, 2\}; \\
& E_{mi}(t) = \sum_{n = -\infty}^{\infty} v^{(i)}_n(t) \exp\pa{\omega \, \textbf{i}\, n \, t}, \hspace{5pt} i \in \{1, 2\};  \\
& I_{mi}(t) = \sum_{n = -\infty}^{\infty} w^{(i)}_n(t) \exp\pa{\omega \, \textbf{i}\, n \, t}, \hspace{5pt} i \in \{1, 2\};  \\
& I_{M}(t) = \sum_{n = -\infty}^{\infty} x_n(t) \exp\pa{\omega \, \textbf{i}\, n \, t};  \\
& I_{N}(t) = \sum_{n = -\infty}^{\infty} y_n(t) \exp\pa{\omega \, \textbf{i}\, n \, t}, \\
\end{aligned}
\end{equation} for complex-valued functions $u^{(1)}_{n}(t), v^{(1)}_{n}(t), w^{(1)}_{n}(t), u^{(2)}_{n}(t), v^{(2)}_{n}(t), w^{(2)}_{n}(t), x_{n}(t), y_{n}(t)$ that tend towards a steady state, such that 
\begin{equation} \begin{aligned}
&\lim_{t \to \infty} \pa{u^{(1)}_{n}(t), v^{(1)}_{n}(t), w^{(1)}_{n}, u^{(2)}_{n}(t), v^{(2)}_{n}(t), w^{(2)}_{n}(t), x_{n}(t), y_{n}(t)} = \\ & \pa{u^{(1)}_{n}, v^{(1)}_{n}, w^{(1)}_{n}, u^{(2)}_{n}, v^{(2)}_{n}, w^{(2)}_{n}, x_{n}, y_{n}} \in \mathbb{C}^{5}. \end{aligned}
\end{equation} 
As $t \to \infty$, the solutions of \eqref{eq:humsimplify} can then be described by a Fourier series with constant coefficients. For sufficiently small $a, b, c$, the right hand sides of \eqref{eq:humsimplify} imply that the nonconstant arguments of each of our eight limiting solutions are small. 

Defining $S^{*}_{m1}(t)$ as the periodic function satisfying $\lim_{t \to \infty} \pa{S^{*}_{m1}(t) - S_{m1}(t)} = 0$, and defining the functions $E^{*}_{m1}, I^{*}_{m1}, S^{*}_{m2}, E^{*}_{m2}, I^{*}_{m2}, I^{*}_M, I^{*}_N$ similarly, we formulate the ansatz that  

\begin{equation}\label{coeff}
\begin{aligned}
& S^{*}_{mi}(t) \approx u^{(i)}_{-1} \exp \pa{-\omega \textbf{i} t} + u^{(i)}_0 + u^{(i)}_{1} \exp \pa{\omega \textbf{i} t} + u^{(i)}_\epsilon, \hspace{5pt} i \in \{1, 2\}; \\
& E^{*}_{mi}(t)  \approx v^{(i)}_{-1} \exp \pa{-\omega \textbf{i} t} + v^{(i)}_0 + v^{(i)}_{1} \exp \pa{\omega \textbf{i} t} + v^{(i)}_\epsilon, \hspace{5pt} i \in \{1, 2\};  \\
& I^{*}_{mi}(t) \approx w^{(i)}_{-1} \exp \pa{-\omega \textbf{i} t} + w^{(i)}_0 + w^{(i)}_{1} \exp \pa{\omega \textbf{i} t} + w^{(i)}_\epsilon, \hspace{5pt} i \in \{1, 2\};  \\
& I^{*}_{M}(t)  \approx x_{-1} \exp \pa{-\omega \textbf{i} t} + x_0 + x_{1} \exp \pa{\omega \textbf{i} t} + x_\epsilon;  \\
& I^{*}_{N}(t)  \approx y_{-1} \exp \pa{-\omega \textbf{i} t} + y_0 + y_{1} \exp \pa{\omega \textbf{i} t} + y_\epsilon, \\
\end{aligned}
\end{equation}
where the Fourier coefficients associated to non-constant terms are much smaller than those associated to constant terms and \be \begin{aligned} f^{(i)}_{\epsilon} << f^{(i)}_1, \hspace{0.2cm} f'_{\epsilon} << f'_1\end{aligned} \ee for $f \in \{u, v, w\}, f' \in \{x, y\}, i \in \{1, 2\}$. We note that the DFE ($u^{(1)}_0 = Q_1, u^{(2)}_0 = Q_2$, all other coefficients equal to zero) is a solution of \eqref{eq:humsimplify}.

In Appendix A, we find a single nonlinear equation \eqref{eqforw2}, the solution set of which contains all solutions for $w_0^{(2)}$ assuming sufficiently-small $\abs{v}$ and sufficiently-large $\omega$. We also outline a method to extract approximations for all other Fourier coefficients in \eqref{coeff}. This yields approximations for the endemic solutions of our system \eqref{eq:humsimplify} given the ansatz in \eqref{coeff}. 

To calculate limiting solutions for $L_N(t)$ and $L_M(t)$ from the approximated solutions for $I^*_{m1}(t), I^*_{m2}(t)$, we note that, by \eqref{livdensityn} and \eqref{livdensity}, 
\be\label{lnsolution} \begin{aligned} L_N(t) \approx N \exp\Bigr[R\pa{f(t), I_{m1}}\Bigr] - N \exp \Bigr[R\pa{g(t)},I_{m1}\Bigr],\end{aligned} \ee and  
\be\label{lmsolution} \begin{aligned} L_M(t) \approx M \exp\Bigr[R\pa{f(t), I_{m1}} + R\pa{f(t), I_{m2}}\Bigr] - M \exp \Bigr[R\pa{g(t), I_{m1}} + R\pa{g(t), I_{m2}}\Bigr].\end{aligned} \ee  
Here, 

\be \begin{aligned} R\pa{r(t), I_{m1}} = \frac{uab\pa{\pa{\int_0^{\infty} r(\tau) d\tau}w^{(1)}_0 + \pa{\int_0^{\infty} e^{-\omega \textbf{i} \tau} r(\tau) d\tau}w^{(1)}_{1}e^{\omega \textbf{i} t} + \pa{\int_0^{\infty} e^{\omega \textbf{i} \tau} r(\tau) d\tau} w^{(1)}_{-1}e^{-\omega \textbf{i} t}}}{N + Mu} \\ 
R\pa{r(t), I_{m2}} = \frac{ab\pa{\pa{\int_0^{\infty} r(\tau) d\tau}w^{(2)}_0 + \pa{\int_0^{\infty} e^{-\omega \textbf{i} \tau} r(\tau) d\tau}w^{(2)}_{1}e^{\omega \textbf{i} t} + \pa{\int_0^{\infty} e^{\omega \textbf{i} \tau} r(\tau) d\tau} w^{(2)}_{-1}e^{-\omega \textbf{i} t}}}{M}.\end{aligned} \ee

The solution set of \eqref{eqforw2} includes the solution $w_0^{(2)} = 0$ (which implies that all other Fourier coefficients but $u_0^{(1)}, u_0^{(2)}$ are zero, and immediately corresponds to the DFE solution of \eqref{eq:humsimplify}). In Appendix B, we find necessary and sufficient conditions for \eqref{eqforw2} to have no endemic solutions. In particular, we obtain the necessary condition 
\be\label{condition1}\begin{aligned}\frac{Mg_2(g_2 + n)}{ a^2 b c n \abs{C_0} Q_2} > 1.\end{aligned} \ee
and the sufficient condition
 \be\label{condition2} \begin{aligned} \frac{Mg_2(g_2 + n)}{Q_2}\pa{\frac{1}{a^2bcn\abs{C_0}} - \frac{Q_1Mu}{{2g_1\pa{g_1 + n}\pa{N+Mu}}}\pa{\frac{u+v}{N + (u + v)M} + \frac{u-v}{N + (u - v)M}}} > 1. \end{aligned} \ee 
We remark on the biological implications of these conditions below.

\begin{rem}\label{rem1} If \[\frac{1}{a^2bcn\abs{C_0}} < \frac{Q_1Mu}{{2g_1\pa{g_1 + n}\pa{N+Mu}}}  \pa{\frac{u+v}{N + (u + v)M} + \frac{u-v}{N + (u - v)M}},\] a malaria epidemic may persist indefinitely solely as a result of the Patch 1 mosquito population. Otherwise, the left hand side of \eqref{condition2} is a monotonically-increasing function of $g_1, g_2, \abs{v}$ and a monotonically-decreasing function of $a, b, c, \abs{v}, Q_1, Q_2$. We notice that, for realistic values of parameters such as those in [5], increases in $g_1$ or decreases in $Q_1$ may be insufficient to prevent disease endemicity, as \[Mg_2(g_2 + n) < a^2 b c n \abs{C_0} Q_2.\] However, for sufficiently-large $g_2$ or sufficiently-small $Q_2$, our approximation for \eqref{eq:humsimplify} has no endemic solutions. \end{rem} 

The conditions \eqref{condition1} and \eqref{condition2} suggest that the $R_0$ for our system satisfies
\be\label{R0bound} \begin{aligned}\frac{a^2 b c n \abs{C_0} Q_2}{Mg_2(g_2 + n)} \le R_0 \le \frac{Q_2}{Mg_2(g_2 + n)\pa{\frac{1}{a^2bcn\abs{C_0}} - \frac{Q_1Mu}{{2g_1\pa{g_1 + n}\pa{N+Mu}}}\pa{\frac{u+v}{N + (u + v)M} + \frac{u-v}{N + (u - v)M}}}}.\end{aligned} \ee

\begin{rem}\label{rem2} The lower bound on the $R_0$ in \eqref{R0bound} is equivalent to the basic reproduction number on Patch $2$ in the case that the forest patch contains a non-moving population of density $M$ (such that the movement between the patches is removed). In other words, if the system does not approach the DFE in a homogeneous population of density $M$, \eqref{eqforw2} possesses endemic solutions. \end{rem} 

We prove Theorem \ref{thm2} in Appendix B. The claim in Remark \ref{rem2} is addressed in Appendix C. 

\subsection{Numerical results}\label{numresults}
Solutions for $I_N(t), I_M(t), I_{m1}(t), I_{m2}(t)$ from \eqref{eq:humsimplify} for parameter values implying slower movement shows peaks in $I_M(t)$ and $I_{m1}(t)$ are closely aligned with troughs in $I_N(t)$ and $I_{m2}(t)$, respectively; oscillations are small and difficult to discern (Figure \ref{Numerical_Solution}). We note that even low-to-moderate proportions of infectious mosquitoes (e.g., $8\%$ of all mosquitoes on Patch 1, $19\%$ of all mosquitoes on Patch 2) can result in very high densities of blood-infected moving individuals: the proportion of infectious moving humans relative to the total human population is approximately $0.3939$, meaning that approximately $98\%$ of $\mathcal{M}$ is infectious for the parameter values in Table \ref{parameters}, in comparison to approximately $17\%$ of $\mathcal{N}$.

 Though the mean period of human movement does not significantly impact the density of infectious individuals, longer periods of movement correspond to greater fluctuations in non-moving human prevalence (Figure \ref{fig6}). As described analytically in Appendix A, the amplitude of oscillations in solutions to \eqref{eq:humsimplify} is a monotonically decreasing function of $\omega$. As a result, if the periodic forcing associated with migration is seasonal, movement patterns may have a significant impact on prevalence at any given time. 
 
 For both $\omega = 1$ day$^{-1}$ and $\omega = 1/120$ day$^{-1}$, the density of blood-infected individuals is far higher in $M$ than in $N$, owing to the difference in FORI experienced by the two subpopulations (Figure \ref{fig6}). Liver infectiousness in $N$ peaks early and recedes as $t \to \infty$; liver infectiousness in $M$ and blood infectiousness in both groups is monotonically increasing.    

We calculated approximate solutions for \eqref{eq:humsimplify} under the parameter values in Table \ref{parameters} by first obtaining the largest solution of \eqref{eqforw2}, at which point the value of $w_0^{(2)}$ was used to find all remaining Fourier coefficients in \eqref{coeff}. Analytical approximations generally correspond to numerical solutions (Figure \ref{numvsapprox}).  We note that the approximate solutions for $I^{*}_M(t)$ and $I_{m2}^*(t)$ are particularly close to the numerical ones (less than $0.015\%$ lower), while the approximated solutions for $I^{*}_N(t)$ and $I_{m1}^{*}(t)$ are less accurate (up to $3\%$ lower). Figure \ref{heatmap}, however, illustrates that the approximation for $I^{*}_N$ grows closer to the numerical solution as $v$ decreases and $u$ increases. 

Comparing the constant-FORI expression for $I^*_M(t)$ in \ref{constantFORIsection} to the numerical solution for $I^*_M(t)$ in \eqref{eq:humsimplify} implies that variation in mosquito prevalence (here generated by human movement) propels and reinforces disease transmission across patches (Figure \ref{constantFORI}). In particular, numerical solutions to \eqref{eq:humsimplify} for $I_M(t), t >> 1$,  were compared with solutions to \eqref{toforisub} (generated by setting the constant proportion of infectious mosquitoes  on Patch 1 to $I^c_{m1} = 0.008$, slightly above the maximum of the   solution for $I^*_{m1}(t)$ in \eqref{eq:humsimplify}, and by setting the constant proportion of infectious mosquitoes on Patch 2 to $I^c_{m2} = 0.089$, above the  solution for $I^*_{m2}(t)$ in \eqref{eq:humsimplify}). This choice of constant FORI ensured that 
\[I^c_{m1}(t), I^c_{m2}(t) > I_{m1}(\tau), I_{m2}(\tau)\] for all  $\tau$, i.e., that the proportion of infectious mosquitoes in the constant-FORI case exceeded that in the non-constant-FORI case for all times. Nonetheless, the constant-FORI solution for $I^{*}_M(t)$ proved significantly smaller than the non-constant FORI solution. This implies that the oscillations in mosquito prevalence, and by extension  migration patterns, must be a cause for the extremely high prevalence in Figure \ref{Numerical_Solution}. 
\begin{figure}[h!]
\hfill 
\center
\includegraphics[width=\textwidth]{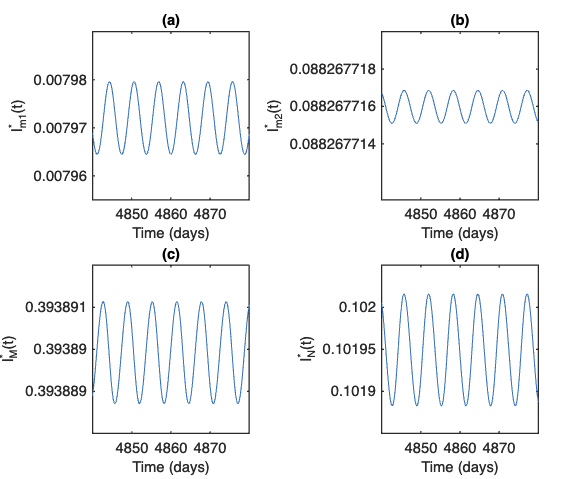}\caption{Solutions to \eqref{eq:humsimplify} using parameter values found in Table \ref{parameters}, with the exception of $\omega$, which was set to $1/5$ days$^{-1}$. Asterisks denote solution asymptotics.   }\label{Numerical_Solution}
\end{figure} 
    
\begin{figure}[h!]
\hfill 
\center
\begin{subfigure}[h!]{\textwidth}
         \centering
         \includegraphics[width=0.8\textwidth]{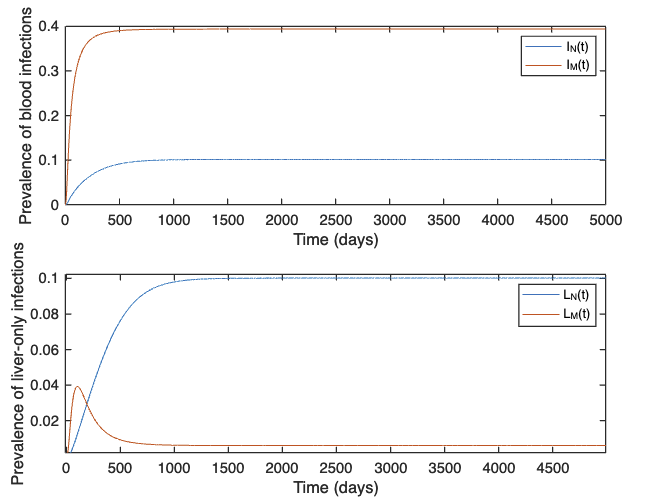}
            \caption{}

         \label{limprev}
     \end{subfigure}
\begin{subfigure}[h!]{0.82\textwidth}
         \centering
         \includegraphics[width=\textwidth]{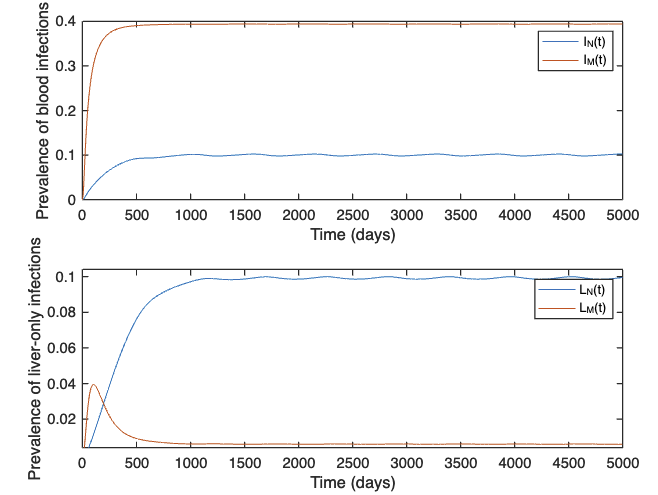}
         \caption{}
         \label{highomega}
     \end{subfigure}
\caption{Prevalence of blood and liver infections  for (a) $\omega = 1$ day$^{-1}$ and (b) $\omega = 1/120$ days$^{-1}$. Numerical solutions were computed using \eqref{eq:humsimplify}.}\label{fig6}\end{figure}

\begin{figure}[h!]
      \centering
     \includegraphics[width=0.8\textwidth]{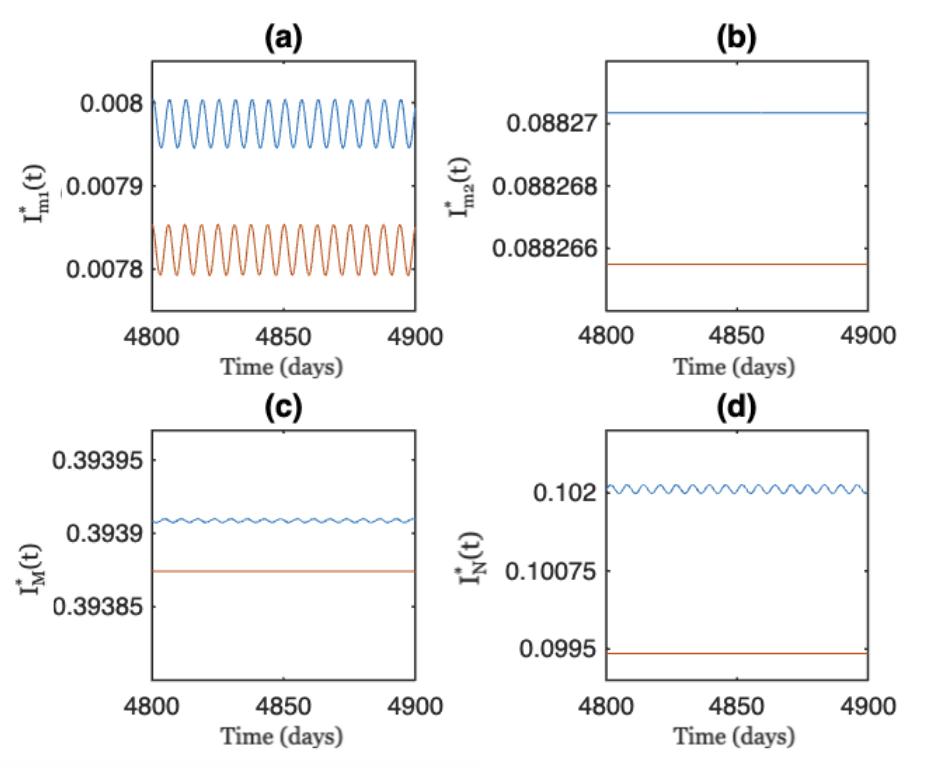}
          \label{approxplot}
      \caption{\small Approximate solutions of \eqref{eqforw2} are compared to numerical solutions of \eqref{eq:humsimplify} for each of (a) $I^*_{m2}(t)$, (b) $I^*_{m1}(t)$, (c)  $I^*_{M}(t)$, and (d) $I^*_N(t)$, using the parameter values in Table \ref{parameters}. In each plot, the blue curve illustrates the numerical solution, while the orange curve illustrates the approximate solution, derived by retrieving the largest solution of \eqref{eqforw2}.}
      \label{numvsapprox}
 \end{figure}

\begin{figure}[h!]
\centering
 \begin{subfigure}[h!]{\textwidth} \centering
\includegraphics[width=0.8\textwidth]{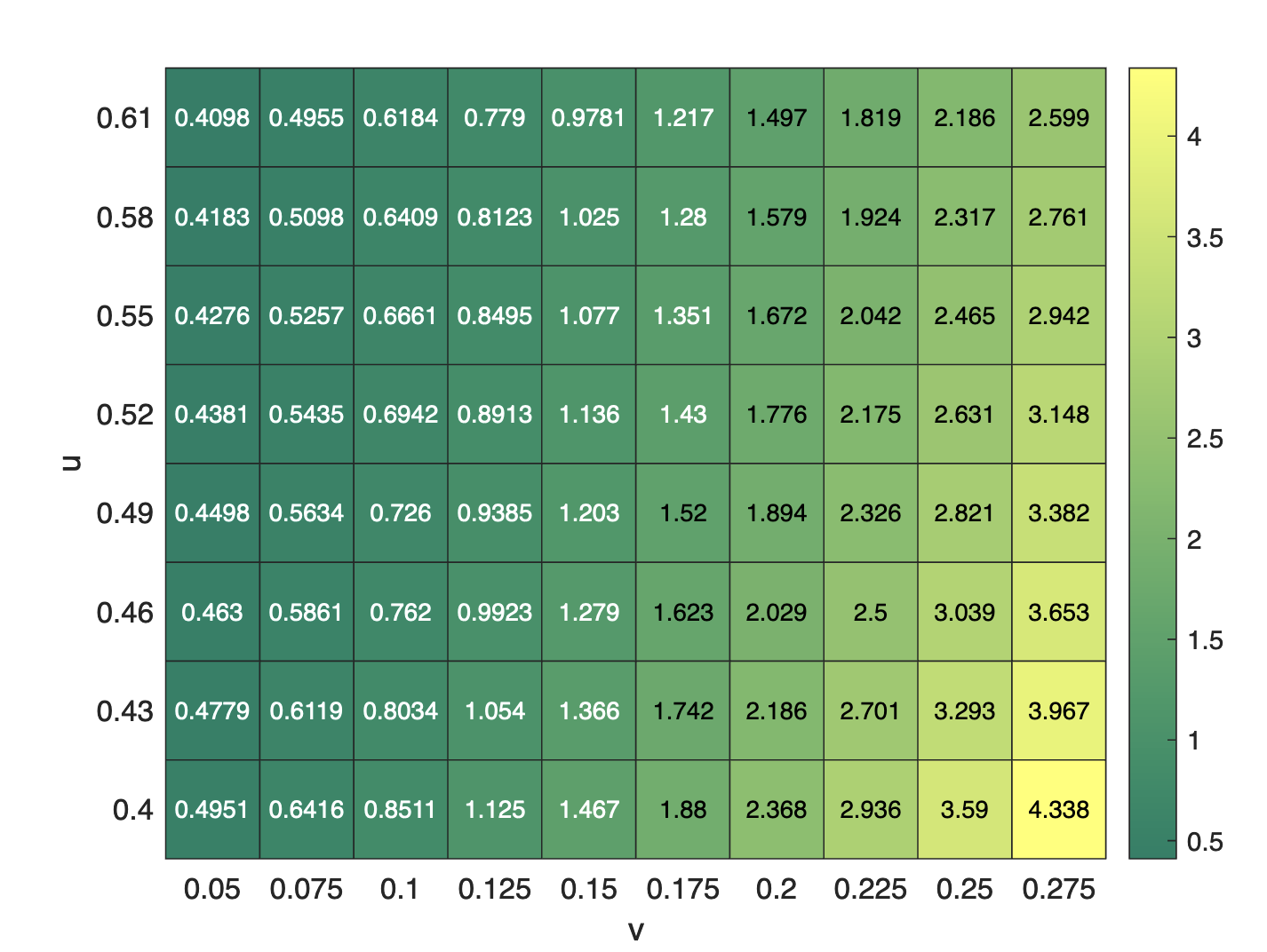} \caption{}

 \end{subfigure}    
    \begin{subfigure}[h!]{\textwidth}
         \centering
         \includegraphics[width=0.6\textwidth]{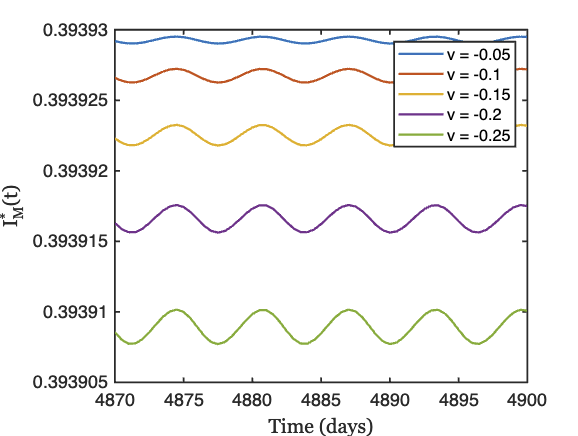}
         \caption{}
         \label{vplot}
     \end{subfigure}
      \caption{\small (a) Percent difference between the numerical solution of \eqref{eq:humsimplify} and the analytical approximation for $I^{*}_N$ (values inside the heatmap) for various values of $u, v$. The approximated solution is derived by first finding the largest solution for $w_0^{(2)}$ of \eqref{eqforw2}, and then using \eqref{forsub1}, \eqref{forsub3} in Appendix A to find $y_0$. (b) Prevalence among moving individuals after $4870$ days for various $v$. Remaining parameters are fixed at the values in Table \ref{parameters}.} \label{heatmap}\end{figure}
\begin{figure}[h!]
\centering
     \includegraphics[width=0.6\textwidth]{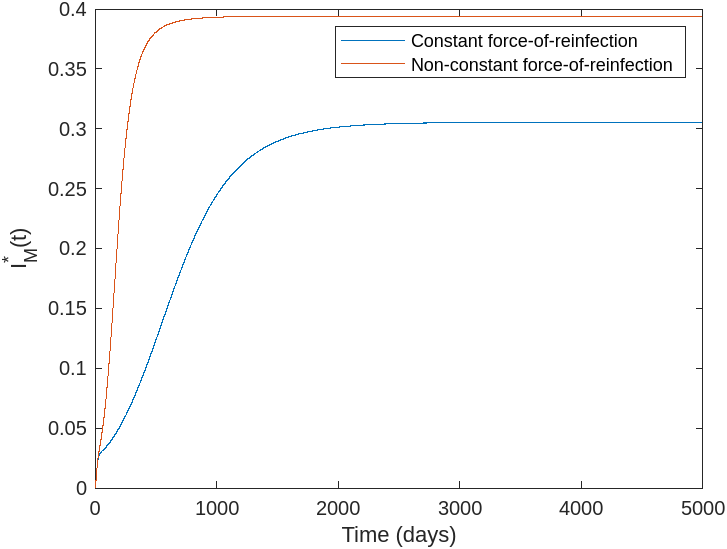}
      \caption{\small Comparison between limiting-state solution for $I_M(t)$ assuming constant FORI (see \eqref{toforisub}, blue line) and non-constant FORI (brown line). Parameter values taken from Table \ref{parameters}; values of $\lambda_1, \lambda_2$ for the constant FORI case selected to slightly exceed the limiting prevalence among mosquito populations on Patches $1$ and $2$, respectively, in the non-constant FORI case. }      \label{constantFORI}\end{figure}
\newpage

\section{Discussion} \label{discussion}

In this paper, we use a three-scale mathematical model -- combining the within-host, population-based, and spatial levels of analysis -- to study the impact of periodic travel patterns on the prevalence of \textit{Plasmodium vivax} infections. Our analysis centers on a two-patch ``village-forest” model characterized by the cyclic migration of a subset of humans from the low-transmission ``village” patch to the high-transmission ``forest” patch. In constructing our model, we further extend the framework developed in \cite{Anwar} and \cite{Mehra}.

We circumvent the need to include direct interactions between compartments of susceptible, liver-infected, and blood-infected humans, by recognizing that the distribution across mosquito states depends only on the densities of \textit{infectious} humans, which can in turn be approximated as functions of the infectious mosquito population using PGFs for within-host dynamics. This observation reduces the complexity of our model, allowing us to compose a ``minimal” system of IDEs, \eqref{eq:humsimplify}, but also leads us to introduce an additional assumption in Section \ref{eq:deriv}, where we suppose that all individuals move near-identically. Embedding the within-host framework from \cite{Mehra} into the model ensures that we account for \textit{superinfection}, in which individuals can possess multiple blood-stage infections simultaneously.

In order to study the model system \eqref{eq:humsimplify}, we introduce the ansatz in Section \ref{approx} that the asymptotic solutions of the system can be closely approximated by rapidly-convergent Fourier series, from where we reduce the original system \eqref{eq:humsimplify} to a single nonlinear equation, \eqref{eqforw2}, in Appendix A. In Appendix B, we determine necessary and sufficient parameter-based conditions for which the only solution to \eqref{eqforw2} corresponds to the DFE, obtaining bounds for the system's basic reproduction number, $R_0$ (see Section \ref{approx}). Though these results capture the impact of movement on the upper and lower bounds for $R_0$, the approximations underlying them are only reasonable for large $\omega$ and small $v$.

To our knowledge, our analysis is the first to incorporate the complex within-host dynamics of hypnozoite inoculation, activation, and clearance characteristic of \textit{P. vivax} into upper and lower bounds for a metapopulation model $R_0$. It may also be possible to extend the system asymptotics analysis in Appendix A and the bounds on the $R_0$ to a larger $n$-patch case, using certain simplifying assumptions on the inter-patch movement (e.g., assumptions on the extent of coupling between limiting equations in the metapopulation model).

\subsection{Biological Insights and Conclusions} \label{bioinsight}
  
Multiple biological applications can be proposed for the model introduced in this paper. Oscillations in malaria prevalence are frequently noted in both human and \textit{Anopheles} populations, but are commonly attributed solely to seasonal fluctuations in mosquito number \cite{grillet2014periodicity}. Our work suggests that, although the fundamental frequency of these oscillations is likely determined by seasonal changes, higher-frequency harmonics in the prevalence may be induced by the cyclic migration of individuals (which may repeat weekly-to-monthly \cite{bannister2019forest}). 

Moreover, though the mean frequency of movement does not impact average \textit{Plasmodium vivax} prevalence, less frequent movement results in greater variation in prevalence over time. We show that individuals traveling between a forest patch with high Anopheles density and a village patch with low Anopheles density see as much risk -- or more -- compared to individuals isolated on the forest patch, which emphasizes the role in movement in exacerbating \textit{P. vivax} epidemics.   

The sufficient and necessary conditions for our approximated equations to have no endemic solutions imply interesting results regarding the impact of parameters on the \textit{P. vivax} basic reproduction number. The upper bound on $R_0$, for instance, is a strictly decreasing function of $\abs{v}$, meaning that a larger value for $\abs{v}$ --- corresponding to more individuals traveling in groups --- may enable the system to reach a disease-free equilibrium. Moreover, the lower bound on $R_0$ have no nonzero solutions is completely independent of $g_1$ (the village mosquito death rate) and $Q_1$ (the village mosquito population size), while the upper bound depends far more strongly on $g_2, Q_2$ (the forest mosquito death rate and population size, respectively) than on $g_1, Q_1$. In the case of \textit{P. vivax} in populations of traveling loggers and farmers (so-called ``forest goers") in Southeast Asia, our work stresses the 
need to control forest mosquito populations in addition to village ones, emphasizing the use of outdoor-acting interventions such as spatial repellents, sterile insect techniques, and source reduction \cite{achee2012spatial, gu2006source, kitron1989suppression, klassen2009introduction}. 

In future work, we plan to extend these results to the location of endemic solutions, considering the impact of mass drug administration and non-pharmaceutical interventions on \textit{P. vivax} prevalence. We also plan to incorporate correlates of immunity in the model, which may alter the clear $\mathcal{M}$-dominated patterns of prevalence visible in Section \ref{numerical_solution} above. 

Addressing the continued endemicity of malaria in regions across Southeast Asia, the Americas, and the Pacific requires a thorough understanding of each of the mechanisms propelling \textit{P. vivax} transmission, and of the interplay between them. As the first paper consolidating representations of the three major scales -- the within-host, population, and metapopulation levels -- at which \textit{P. vivax} epidemiological dynamics occur, we have devised a framework in this paper that can be used to predict \textit{P. vivax} prevalence and to establish conditions for disease eradication given a variety of biological and geographical phenomena. 

\section*{Appendix A: Analytical Approximations for Fourier Coefficients in \eqref{coeff}}\label{proofthm1}
We prove the following result:
\begin{thm}\label{thm1} For sufficiently-small $\abs{v}$ and sufficiently-large $\omega$, the solution for $w^{(2)}_0$ is contained in the solution set of the nonlinear equation \be\label{eqforw2} \begin{aligned} & 
F_1(w_0^{(2)}, A(w_0^{(2)})) + F_2(w_0^{(2)}, A(w_0^{(2)})) = H(A(w_0^{(2)})), \end{aligned} \ee where \be \begin{aligned} & F_1(w_0^{(2)}, A(w_0^{(2)})) = \frac{N - N \exp \pa{\frac{ab\,\abs{C_0}A(w^{(2)}_0)}{N + Mu}} + \frac{M(u + v)g_2(g_2 + n)w^{(2)}_0 }{ac\pa{Q_2  n - (g_2 + n)w_0^{(2)}}}}{N + (u + v)M} \\ & F_2(w_0^{(2)}, A(w_0^{(2)})) = \frac{N - N \exp \pa{\frac{ab\,\abs{C_0}A(w^{(2)}_0)}{N + Mu}} + \frac{M(u - v)g_2(g_2 + n)w^{(2)}_0 }{ac\pa{Q_2  n - (g_2 + n)w_0^{(2)}}}}{N + (u - v)M}, \\ & H(A(w_0^{(2)})) = \frac{2g_1A(w^{(2)}_0)\pa{g_1 + n}}{ac\pa{nQ_1 - (g_1 + n)A(w^{(2)}_0)}}, \\ & A(w^{(2)}_0) = \frac{N + Mu}{ab\,C_0u} \pa{\frac{ab\,\abs{C_0}w^{(2)}_0}{M} + \ln\pa{1 - \frac{g_2(g_2 + n)w^{(2)}_0 }{ac\,\pa{Q_2 n - (g_2 + n)w_0^{(2)}}}}}, \end{aligned} \ee 
\be \begin{aligned} C_0 = \lim_{t \to \infty} \int_0^t f(\tau) \ d\tau, \end{aligned} \ee and a \textit{unique} value for all other Fourier coefficients (and hence limiting-state solutions for \eqref{eq:humsimplify}) can be determined from the value of $w^{(2)}_0$. In particular, when the only solution of \eqref{eqforw2} is $w^{(2)}_0 = 0$, no endemic solution exists. If nonzero solutions of \eqref{eqforw2} exist, then one can extract periodic endemic solutions approximately corresponding to solutions for \eqref{eq:humsimplify}. \end{thm}

\begin{proof}
We first write equations for the Fourier coefficients in \eqref{coeff}. The first six equations in \eqref{eq:humsimplify} can be written as
\begin{equation}\label{1}\begin{aligned} \pa{N + M\pa{u + \frac{v}{2} e^{\omega\textbf{i}\,t} + \frac{v}{2} e^{-\omega\textbf{i}\,t}}} \sum_{k = -1}^{1} \textbf{i} \, j \, u^{(1)}_j e^{\omega\textbf{i}\, j \, t} = \\ \pa{N + M\pa{u + \frac{v}{2} e^{\omega\textbf{i}\,t} + \frac{v}{2} e^{-\omega\textbf{i}\,t}}} \pa{g_1 Q_1 - g_1 \sum_{j = -1}^{1} u^{(1)}_j e^{\omega \textbf{i}\, j \, t}} - \\ ac\, \pa{\sum_{j = -1}^{1} u^{(1)}_j e^{\omega\textbf{i}\, j \, t}}\pa{\sum_{j = -1}^{1} y_j e^{\omega\textbf{i}\, j \, t} + \pa{u + \frac{v}{2} e^{\omega\textbf{i}\,t} +\frac{v}{2} e^{-\omega\textbf{i}\,t}}  \sum_{j = -1}^{1} x_j e^{\omega\textbf{i}\, j \, t}}, \end{aligned} \end{equation} 

\begin{equation}\label{2}\begin{aligned} \pa{N + M\pa{u + \frac{v}{2} e^{\omega\textbf{i}\,t} +\frac{v}{2} e^{-\omega\textbf{i}\,t}}} \sum_{j = -1}^{1} \textbf{i} \, j \, v^{(1)}_j e^{\omega\textbf{i}\, j \, t} = \\ ac\, \pa{\sum_{j = -1}^{1} u^{(1)}_j e^{\omega\textbf{i}\, j \, t}}\pa{\sum_{j = -1}^{1} y_j e^{\omega\textbf{i}\, j \, t} + \pa{u + \frac{v}{2} e^{\omega\textbf{i}\,t} + \frac{v}{2} e^{-\omega\textbf{i}\,t}}  \sum_{j = -1}^{1} x_j e^{\omega\,\textbf{i}\, j \, t}} -\\ \pa{g_1 + n}\pa{N + M\pa{u + \frac{v}{2} e^{\omega\textbf{i}\,t} + \frac{v}{2}e^{-\omega\textbf{i}\,t}}}  \pa{\sum_{j = -1}^{1} v^{(1)}_j e^{\omega\textbf{i}\, j \, t}}, \end{aligned} \end{equation} 

\begin{equation}\label{3}\begin{aligned} \sum_{j = -1}^{1} \textbf{i} \, j \,  w^{(1)}_j e^{\omega\, \textbf{i}\, j \, t} = j \, \sum_{j = -1}^{1} v^{(1)}_j e^{\omega\,\textbf{i}\, j \, t} - g_1 \sum_{j = -1}^{1} w^{(1)}_j e^{\omega\,\textbf{i}\, j \, t},
\end{aligned} \end{equation} 

\begin{equation}\label{4}\begin{aligned} \sum_{j = -1}^{1} \textbf{i} \, j \, u^{(2)}_j e^{\omega\,\textbf{i}\, j \, t} = g_2 Q_2 - \frac{ac}{2M}\pa{\sum_{j = -1}^{1} u^{(2)}_j e^{\omega\,\textbf{i}\, j \, t}}\pa{\sum_{j = -1}^{1} x_j \, e^{\omega\,\textbf{i}\, j \, t}} - \\ g_2 \sum_{j = -1}^{1} u^{(2)}_j e^{\omega\,\textbf{i}\, j \, t},\end{aligned} \end{equation} 

\begin{equation}\label{5}\begin{aligned} \sum_{j = -1}^{1} \textbf{i} \, j \,  v^{(2)}_j e^{\omega\,\textbf{i}\, j \, t} = \frac{ac}{2M}\pa{\sum_{j = -1}^{1} u^{(2)}_j e^{\omega\,\textbf{i}\, j \, t}}\pa{\sum_{j = -1}^{1} x_j \, e^{\omega\,\textbf{i}\, j \, t}} - \\  (g_2 + n) \sum_{j = -1}^{1} v^{(2)}_j e^{\omega\,\textbf{i}\, j \, t}, \end{aligned} \end{equation} 

\begin{equation}\label{6}
    \begin{aligned} 
    \sum_{j = -1}^{1} \textbf{i} \, j \,  w^{(2)}_j e^{\omega\textbf{i}\, n \, t} = n \sum_{j = -1}^{1} v^{(2)}_j e^{\omega\textbf{i}\, j \, t} - g_2 \sum_{j = -1}^{1} w^{(2)}_j e^{\omega\textbf{i}\, j \, t}.
    \end{aligned}
\end{equation}

To deal with the unwieldy equations \eqref{1} through \eqref{6}, we first make the change of parameters \be \begin{aligned} y = y_{-1} + y_{0} + y_{1}; \hspace{5 pt} w_1 =  w^{(1)}_{-1} + w^{(1)}_{0} + w^{(1)}_1; \hspace{5pt} w_2 = w^{(2)}_{-1} + w^{(2)}_{0} + w^{(2)}_1; \hspace{5 pt} v_1 =  v^{(1)}_{-1} + v^{(1)}_{0} + v^{(1)}_1; \\ v_2 = v^{(2)}_{-1} + v^{(2)}_{0} + v^{(2)}_1; \hspace{5 pt} \hspace{5 pt} u_1 =  u^{(1)}_{-1} + u^{(1)}_{0} + u^{(1)}_1; \hspace{5pt} u_2 = u^{(2)}_{-1} + u^{(2)}_{0} + u^{(2)}_1; \hspace{5 pt}
x = x_{-1} + x_{0} +  x_{1}, \end{aligned} \ee  
\be \begin{aligned} y^{*} = y_0 + \textbf{i}(y_{1} - y_{-1}); \hspace{5 pt} w^{*}_1 = w^{(1)}_0 + \textbf{i}(w^{(1)}_1 - w^{(1)}_{-1}); \hspace{5pt} w^{*}_2 =  w^{(2)}_0 + \textbf{i}(w^{(2)}_1 - w^{(2)}_{-1}); \\ v^{*}_1 =  v^{(1)}_0 + \textbf{i}(v^{(1)}_1 - v^{(1)}_{-1});  \hspace{5 pt} v^{*}_2 =  v^{(2)}_0 + \textbf{i}(v^{(2}_1 - v^{(2)}_{-1});   \hspace{5 pt} u^{*}_1 =  u^{(1)}_0 + \textbf{i}(u^{(1)}_1 - u^{(1)}_{-1}); \\ u^{*}_2 =  u^{(2)}_0 + \textbf{i}(u^{(2)}_1 - u^{(2)}_{-1});\hspace{5pt} x^{*} = x_0 + \textbf{i}(x_1 - x_{-1}), \end{aligned} \ee and
 \be \begin{aligned} \bar{y} = y_{-1} - y_{0} - y_{1}; \hspace{5 pt} \bar{w}_1 =  w^{(1)}_{-1} - w^{(1)}_{0} - w^{(1)}_1; \hspace{5pt} \bar{w}_2 = w^{(2)}_{-1} - w^{(2)}_{0} - w^{(2)}_1; \hspace{5 pt} \bar{v}_1 =  v^{(1)}_{-1} - v^{(1)}_{0} - v^{(1)}_1; \\ \bar{v}_2 = v^{(2)}_{-1} - v^{(2)}_{0} - v^{(2)}_1; \hspace{5 pt} \hspace{5 pt} \bar{u}_1 =  u^{(1)}_{-1} - u^{(1)}_{0} - u^{(1)}_1; \hspace{5pt} \bar{u}_2 = u^{(2)}_{-1} - u^{(2)}_{0} - u^{(2)}_1, \hspace{5 pt}
\bar{x} = x_{-1} - x_{0} -  x_{1}. \end{aligned} \ee  

Given this change of parameters, we can substitute the values $t = 0$, $t = \pi$, and $t = \pi/2$ into \eqref{1} through \eqref{6} to obtain 24 equations for the 24 new parameters above.
We then eliminate $u_1, u^{\star}_1, \bar{u}_1, u_2, u^{\star}_2, \bar{u}_2$ and $v_1, v^{\star}_1, \bar{v}_1, v_2, v^{\star}_2, \bar{v}_2$ by using the relations
\be\label{w-to-v}\begin{aligned} v_1 =  \frac{g_1w_1}{n} + \frac{2w^{*}_1 - w_1 - \bar{w}_1}{2n} \\ v^{*}_1 =  \frac{g_1w^{*}_1}{n} - \frac{w_1 - \bar{w}_1}{2n} \\ \bar{v}_1 = \frac{g_1\bar{w}_1}{n} - \frac{2w^{*}_1 - w_1 - \bar{w}_1}{2n},\end{aligned} \ee
and
\be\label{u-to-v} \begin{aligned} u_1 = Q_1 - \frac{(g_1 + n)w_1}{n} - \frac{2w^{*}_1 - w_1 - \bar{w}_1}{2n} \\ u^{*}_1 = Q_1 - \frac{(g_1 + n)w^{*}_1}{n} + \frac{w_1 - \bar{w}_1}{2n} \\ \bar{u}_1 = Q_1 - \frac{(g_1 + n)\bar{w}_1}{n} + \frac{2w^{*}_1 - w_1 - \bar{w}_1}{2n}.\end{aligned} \ee

After applying the change of parameters to \eqref{1} and \eqref{4}, we can solve for $x, x^{\star}, \bar{x}$ and $y, y^{\star}, \bar{y}$ in terms of $w_2, w^{\star}_2, \bar{w}_2$ and $w_1, w^{\star}_1, \bar{w}_1$. We can then rewrite the result in terms of the original Fourier series coefficients for $I_{m1}(t), I_{m2}(t)$. In doing so, we obtain

\be\label{eqx} \begin{aligned} x = \frac{M\pa{g_2(g_2 + n)(w^{(2)}_0 + w^{(2)}_1 + w^{(2)}_{-1}) + (2g_2 + n)(\textbf{i}\cdot w^{(2)}_1 - \textbf{i}\cdot w^{(2)}_{-1}) - w^{(2)}_1 - w^{(2)}_{-1}}}{ac\, \pa{n Q_2 - (g_2 + n)(w^{(2)}_0 + w^{(2)}_1 + w^{(2)}_{-1}) - (\textbf{i}\cdot w^{(1)}_1 - \textbf{i}\cdot w^{(1)}_{-1})}}, \\ x^{\star} = \frac{M\pa{g_2(g_2 + n)(w^{(2)}_0 + (\textbf{i}\cdot w^{(2)}_1 - \textbf{i}\cdot w^{(2)}_{-1})) - (2g_2 + n)(w^{(2)}_{1} + w^{(2)}_{-1}) - (\textbf{i}\cdot w^{(2)}_1 - \textbf{i}\cdot w^{(2)}_{-1})}}{ac\, \pa{n Q_2 - (g_2 + n)(w^{(2)}_0 + (\textbf{i}\cdot w^{(2)}_1 - \textbf{i}\cdot w^{(2)}_{-1})) + w^{(2)}_{1} + w^{(2)}_{-1}}}, \\  \bar{x} = \frac{M\pa{g_2(g_2 + n)(w^{(2)}_0 - w^{(2)}_1 - w^{(2)}_{-1}) - (2g_2 + n)(\textbf{i}\cdot w^{(1)}_1 - \textbf{i}\cdot w^{(1)}_{-1}) + w^{(1)}_1 + w^{(1)}_{-1}}}{ac\, \pa{n Q_2 - (g_2 + n)(w^{(2)}_0 - w^{(2)}_1 - w^{(2)}_{-1}) + (\textbf{i}\cdot w^{(1)}_1 - \textbf{i}\cdot w^{(1)}_{-1})}}, \end{aligned} \ee
and
\be\label{eqy} \begin{aligned} \pa{y + (u + v) x} = \pa{\frac{N + (u + v)M}{ac}}\\\frac{g_1(g_1 + n)(w^{(1)}_0 + w^{(1)}_{-1} + w^{(1)}_1) + (2g_1 + n)(\textbf{i}\cdot w^{(1)}_1 - \textbf{i}\cdot w^{(1)}_{-1}) - ( w^{(1)}_1 + w^{(1)}_{-1}) }{nQ_1 - (g_1 + n)(w^{(1)}_0 + w^{(1)}_{-1} + w^{(1)}_1) - (\textbf{i}\cdot w^{(1)}_1 - \textbf{i}\cdot w^{(1)}_{-1})}, \\
\pa{y^{\star} + u x^{\star}} = \pa{\frac{N + Mu}{ac}}\\\frac{g_1(g_1 + n)(w^{(1)}_0 + \textbf{i}\cdot w^{(1)}_1 - \textbf{i}\cdot w^{(1)}_{-1}) - (2g_1 + n)(w^{(1)}_1 + w^{(1)}_{-1}) - (\textbf{i}\cdot w^{(1)}_1 - \textbf{i}\cdot w^{(1)}_{-1})}{nQ_1 - (g_1 + n)(w^{(1)}_0 + \textbf{i}\cdot w^{(1)}_1 - \textbf{i}\cdot w^{(1)}_{-1}) + w^{(1)}_1 + w^{(1)}_{-1}},\\ 
\pa{\bar{y} + (u - v) \bar{x}} = \pa{\frac{N + (u - v)M}{ac}}\\\frac{g_1(g_1 + n)(w^{(1)}_0 - w^{(1)}_{-1} - w^{(1)}_1) - (2g_1 + n)(\textbf{i}\cdot w^{(1)}_1 - \textbf{i}\cdot w^{(1)}_{-1}) + (w^{(1)}_1 + w^{(1)}_{-1}) }{nQ_1 - (g_1 + n)(w^{(1)}_0 - w^{(1)}_{-1} - w^{(1)}_1) + (\textbf{i}\cdot w^{(1)}_1 - \textbf{i}\cdot w^{(1)}_{-1})}.\end{aligned} \ee

To deal with the remaining two equations in \eqref{eq:humsimplify}, we note that, for sufficiently-small $\abs{v}$ and sufficiently-large $\omega$,

\begin{equation} \begin{aligned}
\lim_{t \to \infty} ab\,\int_0^t\pa{\frac{I_{m1}(\tau)}{N + \pa{u + v \cos\pa{\omega \tau}}M} - \frac{I_{m1}(\tau)}{N + Mu}}\, f(t - \tau) \, d\tau \approx \\ \frac{w^{(1)}_0\,ab\,v\,M}{N + Mu} \, \lim_{t \to \infty} \int_0^t\pa{\frac{ \cos(\omega t)}{N + \pa{u + v \cos\pa{\omega \tau}}M}\, f(t - \tau)} \, d\tau \ll 1
\end{aligned} \end{equation}
and 
\begin{equation} \begin{aligned}
\lim_{t \to \infty} ab\,\int_0^t\pa{\frac{I_{m1}(\tau)\pa{u + v \cos\pa{\omega \tau}}}{N + \pa{u + v \cos\pa{\omega \tau}}M} - \frac{uI_{m1}(\tau)}{N + Mu}}\, f(t - \tau) \, d\tau \approx \\ \frac{w^{(1)}_0\,ab\,v\,N}{N + Mu} \, \lim_{t \to \infty} \int_0^t\pa{\frac{ \cos(\omega t)}{N + \pa{u + v \cos\pa{\omega \tau}}M}\, f(t - \tau)} \, d\tau \ll 1,
\end{aligned} \end{equation}
as $ab$ is small, $\frac{M}{N + Mu} < 1$, and $w^{(1)}_0 < Q_1$ (where $Q_1$ is assumed to be $\ll 1$). 

As such, we can make the zero-order approximation 

\begin{equation}
    \begin{aligned}
     & \lim_{t \to \infty} I_{M}(t) = M \pa{1 - \exp\pa{ab\,\int_{0}^{t}  \pa{\frac{uI_{m1}(\tau) }{N + Mu}  \, + \frac{I_{m2}(\tau)}{M}} f(t - \tau)\, d\tau}}, \\
    & \lim_{t \to \infty} I_{N}(t) =  N \pa{1 - \exp\pa{ab\,\int_{0}^{t}\frac{I_{m1} (t)}{N +  Mu}\, f(t - \tau) \, d \tau}}. \\
    \end{aligned}
\end{equation}
It follows that 

\begin{equation}
\begin{aligned}\label{determiningy}  y_0 + y_{-1} e^{-\textbf{i} \, t} +  y_{1} e^{\textbf{i} \, t} \approx N - N \exp \pa{\frac{ab\,C_0w^{(1)}_0 + ab\,C_1 w^{(1)}_1 e^{\omega\,\textbf{i} \, t} + ab\,C_{-1} w^{(1)}_{-1} e^{-\omega\,\textbf{i} \, t} }{N + Mu}} \\ \approx N - N \exp \pa{\frac{ab\,C_0w_0}{N + Mu}} \pa{1 + \frac{ab\,C_1 w^{(1)}_1 e^{\omega\,\textbf{i} \, t} + ab\,C_{-1} w^{(1)}_{-1} e^{-\omega\,\textbf{i} \, t}}{N + Mu} },\end{aligned}    
\end{equation} 
since by our ansatz $w^{(1)}_1$, $w^{(1)}_{-1} \ll 1$.
Similarly, \begin{equation}\label{determiningx} \begin{aligned} x_0 + x_{-1} e^{-\omega\,\textbf{i} \, t} +  x_{1} e^{\omega\,\textbf{i} \, t} \approx M - M \exp \pa{\frac{uab\,(C_0w^{(1)}_0 + C_1 w^{(1)}_1 e^{\omega\,\textbf{i} \, t} + C_{-1} w^{(1)}_{-1} e^{-\omega\,\textbf{i} \, t}) }{N + Mu}} \\ \times \exp \pa{\frac{ab\,(Cw^{(2)}_0 + C_1 w^{(2)}_1 e^{\omega\,\textbf{i} \, t} + C_{-1} w^{(2)}_{-1} e^{-\omega\,\textbf{i} \, t}) }{M}} \\ \approx M - M \exp \pa{\frac{uab\,C_0w^{(1)}_0}{N + uM}} \exp \pa{\frac{ab\,C_0w^{(2)}_0}{M}} \\ \pa{1 + \frac{ab\,C_1 w^{(2)}_1 e^{\omega\, \textbf{i} \, t} + ab\,C_{-1} w^{(2)}_{-1} e^{-\omega\,\textbf{i} \, t}}{M} + \frac{uab\,(C_1 w_1 e^{\omega\,\textbf{i} \, t} + C_{-1} w_{-1} e^{-\omega\,\textbf{i} \, t})}{N + Mu} }. \end{aligned}
\end{equation}

We note that \eqref{determiningy} and \eqref{determiningx} imply that the amplitude of $I^{*}_M, I^{*}_N$ is a  decreasing function of $\omega$. This can be seen by viewing $C_1, C_{-1}$ as functions of $\omega$, i.e., $C_1 = C_1(\omega), C_{-1} = C_{-1}(\omega)$, and noting that 
 $\frac{d}{d \omega} \abs{C_{-1} + C_1}, \frac{d}{d \omega} \abs{C_{-1} - C_{1}} < 0$. 

From \eqref{determiningy} and \eqref{determiningx}, it follows that 
\be\label{w-to-x} \begin{aligned} x_0 \approx M - M \exp \pa{\frac{uab\,C_0w^{(1)}_0}{N + uM}} \exp \pa{\frac{ab\,C_0w^{(2)}_0}{M}}, \end{aligned} \ee 

\be\label{small1} \begin{aligned}x_1 \approx - M \exp \pa{\frac{uab\,C_0w^{(1)}_0}{N + Mu}} \exp \pa{\frac{ab\,C_0w^{(2)}_0}{M}} \pa{\frac{ab\,C_1 w^{(2)}_1}{M} + \frac{uab\,C_1 w^{(1)}_1}{N + Mu} }, \end{aligned} \ee

and \be\label{small2} \begin{aligned} x_{-1} \approx - M\exp \pa{\frac{uab\,C_0w^{(1)}_0}{N + uM}} \exp \pa{\frac{ab\,C_0w^{(2)}_0}{M}} \pa{\frac{ab\,C_{-1} w^{(2)}_{-1}}{M} + \frac{uab\,C_{-1} w^{(1)}_{-1}}{N + Mu} }, \end{aligned} \ee
while
\be\label{forsub1} \begin{aligned} y_0 \approx N - N \exp \pa{\frac{ab\,C_0w^{(1)}_0}{N + Mu}}, \end{aligned} \ee 

\be\label{31} \begin{aligned}y_1 \approx - N \exp \pa{\frac{ab\,C_0w^{(1)}_0}{N + uM}} \pa{\frac{ab\,C_1 w^{(1)}_1}{N + Mu}} , \end{aligned} \ee
and
\be\label{32} \begin{aligned} y_{-1} \approx - N \exp \pa{\frac{ab\,C_0w^{(1)}_0}{N + uM}} \pa{\frac{ab\,C_{-1} w^{(1)}_{-1}}{N + Mu}}. \end{aligned} \ee

We are now able to approximate the constant terms in each of our Fourier series. From \eqref{eqx} and \eqref{eqy}, we have

\be\label{forsub2} \begin{aligned} x_0 \approx \frac{Mg_2(g_2 + n)w^{(2)}_0 }{ac\pa{Q_2  n - (g_2 + n)w_0^{(2)}}}. \end{aligned} \ee

This implies that \be\label{endpoint}\begin{aligned} 0 \le w^{(2)}_0 \le \frac{a c n Q_2}{\pa{g_2 + ac}\pa{g_2 + n}}.\end{aligned} \ee 

Furthermore, \be\label{forsub3} \begin{aligned} w^{(1)}_0 \approx \frac{N + Mu}{ab\,C_0u} \pa{\frac{ab\,
\abs{C_0}w^{(2)}_0}{M} + \ln\pa{1 - \frac{g_2(g_2 + n)w^{(2)}_0 }{ac\,\pa{Q_2 n - (g_2 + n)w_0^{(2)}}}}}.\end{aligned} \ee

We note here that, for realistic $C_0$ such that $ab\,C_0 > 1$, 
$w_0^{(1)}$ is a monotonically-increasing function of $w_0^{(2)}$.

Moreover,
\be\label{tosub} \begin{aligned} \frac{y_0 + (u + v) x_0}{N + (u + v) M} + \frac{y_0 + (u-v) x_0}{N + (u-v) M} \approx \frac{2g_1\pa{g_1 + n}w^{(1)}_0}{ac\pa{nQ_1 - (g_1 + n)w^{(1)}_0}}.\end{aligned} \ee 

Substituting \eqref{forsub1}, \eqref{forsub2}, and \eqref{forsub3} into \eqref{tosub}, we obtain a single nonlinear equation  for $w^{(2)}_0$, arriving at the equation in the statement of Theorem \ref{thm1}.
For each solution $w^{(2)}_0$ of \eqref{eqforw2}, we can obtain a unique solution for the constant term in every other Fourier series: \eqref{forsub3} yields $w^{(2)}_0$, \eqref{w-to-v} and \eqref{u-to-v} yield the constant terms in the Fourier series of $S^{*}_{m1}, S^{*}_{m2}, E^{*}_{m1}$, and $E^{*}_{m2}$, and \eqref{w-to-x} and \eqref{forsub1} yields $x_0$ and $y_0$. 

Given an approximation for the constant terms of the Fourier series, we now outline an analytical method for \textit{uniquely} approximating the coefficients associated with the \textit{non-constant} terms.

We first make the change of variables \[w^{(2)}_s:= w^{(2)}_1 + w^{(2)}_{-1}, \hspace{5pt} w^{(2)}_d:= \textbf{i} \cdot \pa{w^{(2)}_1 - w^{(2)}_{-1}}\] for simplicity, and consider \eqref{eqx}, which becomes
\be\label{eqxrewritten} \begin{aligned} \frac{g_2(g_2 + n)(w^{(2)}_0 + w^{(2)}_s) + (2g_2 + n)(w^{(2)}_d) - w^{(2)}_s}{ n Q_2 - (g_2 + n)(w^{(2)}_0 + w^{(2)}_s) - (w^{(2)}_d)} = \\  \frac{g_2(g_2 + n)(w^{(2)}_0 + w^{(2)}_d) - (2g_2 + n)(w^{(2)}_s) - (w^{(2)}_d)}{n Q_2 - (g_2 + n)(w^{(2)}_0 + w^{(2)}_d) + w^{(2)}_s} = \\ \frac{g_2(g_2 + n)(w^{(2)}_0 - w^{(2)}_s) - (2g_2 + n)(w^{(2)}_d) + w^{(2)}_s}{ n Q_2 - (g_2 + n)(w^{(2)}_0 - w^{(2)}_s) + (w^{(2)}_d)}. \end{aligned} \ee

Now considering \eqref{small1} and \eqref{small2}, we note that, since \[
\exp \pa{\frac{uab\,C_0w^{(1)}_0}{N + Mu}} \exp \pa{\frac{ab\,C_0w^{(2)}_0}{M}} < 1,
\] and $C_1, C_{-1} \ll 1$ for sufficiently-small $\omega$, \[x_1 \ll w^{(1)}_1, w^{(2)}_1 \ll 1, \hspace{10pt} x_{-1} \ll w^{(1)}_{-1}, w^{(2)}_{-1} \ll 1.\]

Hence, we can approximate $w^{(2)}_s, w^{(2)}_d \in \mathbb{R}$ by a real solution of the system of two quadratic equations following from \eqref{eqxrewritten}. This system can be rewritten as a single quartic for $w^{(2)}_s$, which can be shown to possess two real roots by Descartes' rule of signs, where one is extraneous (a solution $\{w^{(2)}_{se}, w^{(2)}_{de}\}$ satisfying $n Q_2 - (g_2 + n) (w^{(2)}_0 + w^{(2)}_{de}) + w^{(2)}_{se} = 0$).

Once we have obtained $w^{(2)}_1, w^{(2)}_{-1}$, we can determine  $w^{(1)}_{1}, w^{(1)}_{-1}$ similarly, using the two quadratic equations formed by substituting \eqref{31} and \eqref{32} into \eqref{eqy} and the previously-determined values for $w^{(1)}_0$ and $y_0$. From here, the remaining coefficients follow. 

\end{proof}
The derivation preceding \eqref{endpoint} in Section \ref{proofthm1} indicates that all (biologically-realistic) solutions for $w^{(2)}_0$ are in the range $0 < w^{(2)}_0 < \frac{a c n Q_2}{\pa{g_2 + ac}\pa{g_2 + n}}$. We now turn to analyzing when \eqref{eqforw2} has \textit{endemic} solutions in the interval above. 

\section*{Appendix B: Necessary and Sufficient Conditions for \eqref{eqforw2} To Have No Endemic Solutions}
We prove the following result: \begin{thm}\label{thm2}
For sufficiently-small $\abs{v}$ and sufficiently-large $\omega$, a necessary condition for \eqref{eqforw2} to have no nonzero solutions is
  \be\begin{aligned}\frac{Mg_2(g_2 + n)}{ a^2 b c n \abs{C_0} Q_2} > 1.\end{aligned} \ee
A stronger, sufficient condition is
 \be\begin{aligned} \frac{Mg_2(g_2 + n)}{Q_2}\pa{\frac{1}{a^2bcn\abs{C_0}} - \frac{Q_1Mu}{{2g_1\pa{g_1 + n}\pa{N+Mu}}}\pa{\frac{u+v}{N + (u + v)M} + \frac{u-v}{N + (u - v)M}}} > 1. \end{aligned} \ee \end{thm}

\begin{proof}[Proof of Theorem \ref{thm2}] We begin with the proof of the sufficient condition. 
Since 
\be \begin{aligned} & F(W) := \frac{N - N \exp \pa{\frac{ab\,\abs{C_0}A(W)}{N + Mu}} + \frac{M(u + v)g_2(g_2 + n)W }{ ac\,\pa{Q_2  n - (g_2 + n)W}}}{N + (u + v)M} \ + \\ & \frac{N - N \exp \pa{\frac{ab\,\abs{C_0}A(W)}{N + Mu}} + \frac{M(u - v)g_2(g_2 + n)W }{ac\,\pa{Q_2  n - (g_2 + n)W}}}{N - (u + v)M} -   \frac{2g_1 A(W)\pa{g_1 + n}}{acnQ_1} > 0\end{aligned} \ee 
when $W$ is a non-negative solution of \eqref{eqforw2}, two sufficient conditions for \eqref{eqforw2} to have no nontrivial roots are 

\be \begin{aligned} & \frac{N - N \exp \pa{\frac{ab\,C_0A(W)}{N + Mu}}}{N + (u + v)M} \ + \frac{N - N \exp \pa{\frac{ab\,\abs{C_0}A(W)}{N + Mu}} }{N + (u - v)M} < 0; \\ & \frac{M(u + v)g_2(g_2 + n)W}{ac(N + (u + v)M)\pa{Q_2  n - (g_2 + n)W}}\ + \\ & \frac{M(u - v)g_2(g_2 + n)W }{ac\pa{N - (u + v)M}\pa{Q_2 n - (g_2 + n)W}} - \frac{2g_1 A(W)\pa{g_1 + n}}{acnQ_1} < 0 \end{aligned} \ee

for all $W \in (0, \frac{a c n Q_2}{\pa{g_2 + ac}\pa{g_2 + n}})$.

Since $A(0) = 0$, the first condition holds if $A'(W) > 0$. The second condition holds if 

\be \begin{aligned} 
\frac{2g_1(g_1 + n)}{a c n Q_1}A'(W) > \pa{\frac{u+v}{N + (u+v)M} + \frac{u-v}{N + (u-v)M}}\\ \pa{\frac{Mg_2(g_2 + n)}{ac(Q_2n - (g_2 + n)W)} + \frac{Mg_2(g_2 + n)^2W}{ac(Q_2n - (g_2 + n)W)^2}},
\end{aligned} \ee
meaning that \be \begin{aligned} & \frac{2g_1(g_1 + n)}{acnQ_1}\frac{N + Mu}{ab\,C_0 u}\pa{\frac{ab\,\abs{C_0}}{M} - \frac{g_2\pa{g_2 + n} \pa{Q_2  n - (g_2 + n)W} + g_2(g_2 + n)^2W}{\pa{Q_2  n - (g_2 + n)W} \pa{a c n Q_2 - W\pa{g_2 + ac}\pa{g_2 + n}}}}  > \\ & \pa{\frac{M(u + v)}{(N + (u + v)M)}\ + \frac{M(u - v)}{N - (u + v)M}} \pa{\frac{g_2(g_2 + n)}{ac(nQ_2 - (g_2 + n)W)}}\pa{1 + \frac{(g_2 + n)W}{nQ_2 - (g_2 +n)W}}.
\end{aligned} \ee

This brings us to the relation
\be \begin{aligned}  & -nQ_2g_2(g_2 + n)\pa{\frac{M(u + v)}{(N + (u + v)M)} + \frac{M(u - v)}{N - (u + v)M}}   + 
 \frac{2g_1(g_1 + n)}{nQ_1}\frac{\pa{N + Mu}}{ab\,C_0 u}\\ & \pa{\frac{ab\,\abs{C_0}(nQ_2 - (g_2 + n)W)^2}{M} - \frac{g_2\pa{g_2 + n} \pa{Q_2  n - (g_2 + n)W} + g_2(g_2 + n)^2W}{ac}} > 0.
 \end{aligned} \ee

Since the left hand side is a monotonically-increasing function of $W$, we obtain the sufficient condition
\be \begin{aligned}  & -nQ_2g_2(g_2 + n)\pa{\frac{M(u + v)}{(N + (u + v)M)} + \frac{M(u - v)}{N - (u + v)M}}   + \\
& \frac{2g_1(g_1 + n)}{nQ_1}\frac{\pa{N + Mu}}{ab\, C_0 u}\pa{\frac{ab\, \abs{C_0}(nQ_2)^2}{M} - \frac{g_2\pa{g_2 + n} \pa{Q_2  n}}{ac}} > 0,
 \end{aligned} \ee
from where condition \eqref{condition2} in Theorem \ref{thm2} follows. 

We now address the necessary condition. If \be \begin{aligned}  F(W):= -\frac{N - N \exp \pa{\frac{ab\,\abs{C_0}A(W)}{N + Mu}}}{N + (u + v)M}  - \frac{N - N \exp \pa{\frac{ab\,\abs{C_0}A(W)}{N + Mu}}}{N + (u - v)M} + \\ \frac{2g_1A(W)\pa{g_1 + n}}{ac\pa{nQ_1 - (g_1 + n)A(W})} < 0 \end{aligned} \ee for some $0 < W < \frac{a c n Q_2}{\pa{g_2 + ac}\pa{g_2 + n}}$, then, since \be \begin{aligned} G(W):= \frac{M(u + v)g_2(g_2 + n)W}{ac\pa{N + (u + v)M}\pa{Q_2  n - (g_2 + n)W}}  + \\ \frac{M(u - v)g_2(g_2 + n)W}{ac\pa{N + (u - v)M}\pa{Q_2  n - (g_2 + n)W}} > 0 \end{aligned} \ee for all $0 < W \le \frac{a c n Q_2}{\pa{g_2 + ac}\pa{g_2 + n}}$ and $F(W) > G(W)$ for $W$ such that \[ 0<\frac{a c n Q_2}{\pa{g_2 + ac}\pa{g_2 + n}} - W << 1,\] $F(W)$ and $G(W)$ must intersect at a point $0 < W < \frac{a c n Q_2}{\pa{g_2 + ac}\pa{g_2 + n}}$. 

As a result, if $A(W) < 0$ for some $W \in (0, \frac{a c n Q_2}{\pa{g_2 + ac}\pa{g_2 + n}})$, \eqref{eqforw2} must possess a positive, biologically-realistic solution for $w^{(2)}_0$. Given that $A''(W) > 0$ for all $W$ in the interval in question, a necessary condition for \eqref{eqforw2} to have no solutions in the interval $(0, \frac{a c n Q_2}{\pa{g_2 + ac}\pa{g_2 + n}})$ is the condition $A'(W) > 0$, from where \eqref{condition1} follows. \end{proof}

\section*{Appendix C: $R_0$ in the case of an isolated Patch $2$ population of density $M$}\label{proofthm3}
In this case, the system \eqref{eq:humsimplify} becomes:
\be\label{neweqhum} \begin{aligned}
    &\frac{d S_{m2}}{dt} = g_{2}Q_2 - \frac{ ac\,I_M}{M} S_{m2} - g_{2} S_{m2}; \\
    & \frac{d E_{m2}}{dt} = \frac{ ac\,I_M}{M} S_{m2} - (g_{2} + n)E_{m2}; \\
    & \frac{d I_{m2}}{dt} = n E_{m2} - g_{2} I_{m2};\\
    & I_{M} = M  - M \exp\pa{\int_{0}^{t}  \pa{\frac{ab\, I_{m2}(\tau)}{M}} f(t - \tau)\, d\tau}; 
    \end{aligned} \ee
with initial conditions \[S_{m2}(0) = S^{(0)}_{m2}; \ E_{m2}(0) = E^{(0)}_{m2};  \ I_{m2}(0) = I^{(0)}_{m2}; \ I_M(0) = 0.\] 

Letting $I_{m2, t} = I_{m2}(t - \phi)$, $\phi \in [-t, 0]$ (such that $I_{m2, t}$ is the \textit{history} of $I_{m2}$) we find that
\[\begin{pmatrix} \frac{d}{dt}  S_{m2} \\ \frac{d}{dt}  E_{m2} \\ \frac{d}{dt} I_{m2} \end{pmatrix} = \textbf{J}(S_{m2}(t), E_{m2}(t), I_{m2, t}), \]
where \[\textbf{J}(S_{m2}, E_{m2}, I_{m2, t}) = \begin{pmatrix} g_{2}Q_2 - acS_{m2} + ac \exp\pa{\int_{0}^{t}  \pa{\frac{ab\, I_{m2}(\tau)}{M}} f(t - \tau)\, d\tau} S_{m2} - g_{2} S_{m2} \\ acS_{m2} - ac \exp\pa{\int_{0}^{t}  \pa{\frac{ab\, I_{m2}(\tau)}{M}} f(t - \tau)\, d\tau} S_{m2} - (g_{2} + n)E_{m2} \\ n E_{m2} - g_{2} I_{m2} \end{pmatrix}.\]

The Fréchet derivative of $\textbf{J}$ evaluated at the DFE is 
\be \begin{aligned} DQ(Q_2, 0, 0)\begin{pmatrix}S \\ E \\ I\end{pmatrix} = \begin{pmatrix} -g_2S + acQ_2\,\pa{\int_{0}^{t}  \pa{\frac{ab}{M}} f(t - \tau)\, d\tau}I  \\ -\pa{g_2 + n}E - acQ_2\,\pa{\int_{0}^{t}  \pa{\frac{ab}{M}} f(t - \tau)\, d\tau}I \\ nE - g_2I\end{pmatrix}.\end{aligned} \ee By Theorem 4.7 in Diekmann and Gyllenberg \cite{diekmann2012equations}, the characteristic equation of \eqref{neweqhum} is equivalent to 
\[\det \begin{pmatrix} -g_2 -\lambda & 0 & acQ_2 \frac{abC_0}{M} \\ 0 & -g_2 - n - \lambda & -acQ_2 \frac{abC_0}{M} \\ 0 & n & -g_2-\lambda \end{pmatrix} = 0,\] where $C_0:= \lim_{t \to \infty} \int_0^t F(\tau) d\tau$ is defined in \eqref{C0}.
The DFE is locally exponentially stable (which corresponds to the $R_0 < 1$ case) if the characteristic equation has all negative roots, and unstable otherwise (which corresponds to the $R_0 > 1$ case). Simplifying the left hand side and using $C_0 < 0$, we obtain \[\pa{-g_2 - \lambda}\pa{(-g_2 - n - \lambda)(-g_2 - \lambda) - nacQ_2 \frac{ab\abs{C_0}}{M}} = 0,\] which yields that the characteristic equation has all negative roots if and only if 
\[\frac{a^2bcn\abs{C_0}Q_2}{Mg_2\pa{g_2 + n}} < 1.\] It follows that \[R_0 = \frac{a^2bcn\abs{C_0}Q_2}{Mg_2\pa{g_2 + n}} = \underbrace{\pa{ab}}_{\substack{\text{\tiny infectious} \\ \text{\tiny bite rate per} \\ \text{\tiny infectious}\\ \text{\tiny mosquito}}} \cdot \underbrace{\frac{Q_2}{M}}_{\substack{\text{\tiny mosquito} \\ \text{\tiny to human} \\ \text{\tiny ratio}}} \cdot \underbrace{\abs{C_0}}_{\substack{\text{\tiny expected} \\ \text{\tiny length of} \\ \text{\tiny infection}\\ \text{\tiny after 1 bite}}} \cdot \underbrace{\pa{ac}}_{\substack{\text{\tiny infecting} \\ \text{\tiny bite rate per} \\ \text{\tiny susceptible}\\ \text{\tiny mosquito}}} \cdot \underbrace{\frac{1}{g_2 + n}}_{\substack{\text{\tiny expected} \\ \text{\tiny time mosquito} \\ \text{\tiny exposed}}} \cdot \underbrace{\frac{n}{g_2}}_{\substack{\text{\tiny sporogyny to} \\ \text{\tiny death rate} \\ \text{\tiny ratio}}} \]
We note that the same result could be obtained by retracing Appendices A and B with $u, v$ both set to zero (in this case, the lower and upper bounds on $R_0$ in \eqref{R0bound} align).
\section*{Funding}
 J.A. Flegg’s research is supported by the Australian Research Council (DP200100747, FT210100034) and the National Health and Medical Research Council (APP2019093). 
\section*{Data availability}
Data sharing not applicable to this article as no datasets were generated or analysed during the
current study.
\appendix

\FloatBarrier
\bibliographystyle{plain}
\bibliography{bibliography}

\end{document}